\documentclass[aps,prb,twocolumn,showpacs]{revtex4}

\usepackage{epsfig}
\usepackage{graphicx}
\usepackage{amsfonts}
\usepackage[figuresright]{rotating}
\usepackage{amssymb}
\usepackage{amsmath}
\usepackage{psfrag}
\usepackage{mdsymbol}

\newtheorem{theorem}{Theorem}
\newtheorem{lemma}{Lemma}

\newenvironment{proof}[1][Proof]{\noindent \textbf{#1: }}{\ \rule{0.5em}{0.5em}}

\def\be{\begin{equation}}
\def\ee{\end{equation}}
\def\bea{\begin{eqnarray}}
\def\eea{\end{eqnarray}}

\def\nn{\nonumber}

\newcommand{\ket}[1]{| #1 \rangle}

\begin{document}
\title{Quaternion, harmonic oscillator, and high-dimensional
topological states }
\author{Congjun Wu}
\affiliation{Department of Physics, University of California, San Diego,
La Jolla, CA 92093}

\begin{abstract}
Quaternion, an extension of complex number, is the first discovered non-commutative division algebra by William Rowan Hamilton in 1843.
In this article, we review the recent progress on building up
the connection between the mathematical concept of quaternoinic analyticity
and the physics of high-dimensional topological states.
Three- and four-dimensional harmonic oscillator wavefunctions
are reorganized by the SU(2) Aharanov-Casher gauge potential to
yield high-dimensional Landau levels possessing the full rotational
symmetries and flat energy dispersions.
The lowest Landau level wavefunctions exhibit quaternionic analyticity, satisfying the {\it Cauchy-Riemann-Fueter} condition, which
generalizes the two-dimensional complex analyticity to three and four dimensions.
It is also the Euclidean version of the helical Dirac and the
chiral Weyl equations.
After dimensional reductions, these states become two- and three-dimensional
topological states maintaining time-reversal symmetry but exhibiting
broken parity.
We speculate that quaternionic analyticity can provide a guiding
principle for future researches on high-dimensional interacting
topological states.
Other progresses including high-dimensional Landau levels
of Dirac fermions, their connections to high energy physics,
and high-dimensional Landau levels in the Landau-type gauges, are also reviewed.
This research is also an important application of the
mathematical subject of quaternion analysis in theoretical
physics, and provides useful guidance for the experimental
explorations on novel topological states of matter.
\end{abstract}
\maketitle

\section{Introduction}
\label{sect:intro}

Quaternions, also called Hamilton numbers, are the first non-commutative
division algebra as a natural extension to complex numbers
(See the quaternion plaque in Fig. \ref{fig:quaternion}).
Imaginary quaternion units $i,j$ and $k$
are isomorphic to the anti-commutative SU(2) Pauli matrices
$-i\sigma_{1,2,3}$.
Hamilton used quaternions to represent three-(3D) and
four-dimensional (4D) rotations, and performed the product of two rotations.
In fact, it is amazing that he was well ahead of his time --
equivalently he was using the spin-$\frac{1}{2}$ fundamental
representations of the SU(2) group, which was before
quantum mechanics was discovered.
Nevertheless, the development of quaternoinic analysis met
significant difficulty since quaternions do not commute.
An important progress was made by Fueter in 1935 as reviewed in
Ref. \cite{sudbery1979}, who defined the Cauchy-Riemann-Fueter condition
for quaternionic analyticity.
Amazingly again, this is essentially the Euclidean version of the
Weyl equation proposed in 1929.
Later on, there have been considerable efforts in constructing
quantum mechanics and quantum field theory based on quaternions
\cite{adler1995,finkelstein1962,yang2005}.

On the other hand, the past decade has witnessed a tremendous
progress in the study of topological states of matter, in particular,
time-reversal invariant topological insulators in two dimensions
(2D) and 3D.
Topological properties of their band structures are characterized
by a $\mathbb{Z}_{2}$-index, which are stable against time-reversal
invariant perturbations and weak interactions
\cite{bernevig2006a,kane2005,kane2005a,fu2007,fu2007a,moore2007,bernevig2006,
wu2006,qi2008,roy2009,roy2010}.
These studies are further developments of quantum anomalous Hall insulators characterized by the integer-valued Chern numbers \cite{thouless1982,haldane1988}.
Later on, topological states of matter including both insulating and
superconducting states have been classified into ten different
classes in terms of their properties under the chiral, time-reversal, and particle-hole symmetries \cite{kitaev2009,schnyder2008}.
These studies have mostly focused on lattice systems.
The wavefunctions of the Bloch bands are complicated, and their energy
spectra are dispersive, both of which are obstacles for the study of
high-dimensional fractional topological states.

In contrast, the 2D quantum Hall states \cite{klitzing1980,tsui1982}
are early examples of topological states of matter studied
in condensed matter physics.
They arise from Landau level quantizations due to the cyclotron
motion of electrons in magnetic fields \cite{girvin1999}.
Their wavefunctions are simple and elegant, which are basically
harmonic oscillator wavefunctions.
They are reorganized
to exhibit analytic properties by an external magnetic field.

Generally speaking, a 2D quantum mechanical wavefunction
$\psi(x,y)$ is complex-valued, but not necessarily complex analytic.
We do not need all the set of 2D harmonic oscillator wavefunctions,
but would like to select a subset of them with non-trivial
topological properties, then complex analyticity is a natural
selection criterion.
Indeed, the lowest Landau level wavefunctions exhibit complex
analyticity.
Mathematically, it is imposed by the Cauchy-Riemann condition
(See Eq. \ref{eq:cauchy} in the text.), and
physically it is implemented by the magnetic field, which reflects
that the cyclotron motion is chiral.
This fact greatly facilitated the construction of Laughlin
wavefunction in the study of fractional quantum Hall states
\cite{laughlin1983}.

How to generalize Landau levels to 3D and even higher dimensions
is a challenging question.
A pioneering work was done by Shoucheng and his former student
Jiangping Hu in 2001 \cite{zhang2001}.
They constructed the Landau level problem on a compact space of
the $S^4$ sphere, which generalizes Haldane's formulation of the 2D
Landau levels on an $S^2$ sphere.
Haldane's construction is based on the 1st Hopf map \cite{haldane1983},
in which a particle is coupled to the vector potential from
a $U(1)$ magnetic monopole.
Zhang and Hu considered a particle lying on the $S^4$ sphere coupled to
an SU(2) monopole gauge field, and employed the 2nd Hopf map
which maps a unit vector on the $S^4$ sphere to a normalized
4-component spinor.
The Landau level wavefunctions are expressed in terms of the
four components of the spinor.
Such a system is topologically non-trivial characterized by
the 2nd Chern number possessing time-reversal symmetry.
This construction is very beautiful, nevertheless, it needs
significantly advanced mathematical physics knowledge which
may not be common for the general readers in the condensed
matter physics, and atomic, molecular, and optical physics
community.

We have constructed high-dimensional topological states
(e.g. 3D and 4D) based on harmonic oscillator wavefunctions
in flat spaces \cite{li2012a,li2013}.
They exhibit flat energy dispersions and non-trivial topological
properties, hence, they are generalizations of the 2D Landau level
problem to high dimensions.
Again we will select and reorganize a subset of wavefunctions
in seeking for non-trivial topological properties.
The strategy we employ is to use quaternion analyticity
as the new selection criterion to replace the previous one of
complex analyticity.
Physically it is imposed by spin-orbit coupling, which
couples orbital angular momentum and spin together
to form the helicity structure.
In other words, the helicity generated by spin-orbit
coupling plays the role of 2D chirality due to the
magnetic field.
Our proposed Hamiltonians can also be formulated in terms of
spin-$\frac{1}{2}$ fermions coupled to the SU(2) gauge potential,
or, the Aharanov-Casher potential.
Gapless helical Dirac surface modes, or, chiral Weyl modes,
appear on open boundaries manifesting the non-trivial
topology of the bulk states.

We have also constructed high-dimensional Landau levels of Dirac
fermions \cite{li2012}, whose Hamiltonians can be interpreted in
terms of complex quaternions.
The zeroth Landau levels of Dirac fermions are a branch of
half-fermion Jackiw-Rebbi modes \cite{jackiw1976}, which are
degenerate over all the 3D angular momentum quantum numbers.
Unlike the usual parity anomaly and chiral anomaly in which
massless Dirac fermions are minimally coupled to the
background gauge fields,
these Dirac Landau level problems correspond to a non-minimal
coupling between massless Dirac fermions and background fields.
This problem lies at the interfaces among condensed matter physics,
mathematical physics, and high energy physics.

High-dimensional Landau levels can also be constructed in the
Landau-type gauge, in which rotational symmetry is explicitly
broken \cite{li2013a}.
The helical, or, chiral plane-waves are reorganized by
spatially dependent spin-orbit coupling to yield non-trivial
topological properties.
The 4D quantum Hall effect of the SU(2) Landau levels have
also been studied in the Landau-type gauge, which exhibits
the quantized non-linear electromagnetic response as a spatially
separated 3D chiral anomaly.

We speculate that quaternionic analyticity would act as a guiding
principle for studying high-dimensional interacting topological
states, which are a major challenging question.
The high-dimensional Landau level problems reviewed below provide
an ideal platform for this research.
This research is at the interface between mathematical and
condensed matter physics, and has potential benefits
to both fields.

This review is organized as follows:
In Sect. \ref{sect:history}, histories of complex number and
quaternion, and the basic knowledge of complex analysis and
quaternion analysis are reviewed.
In Sect. \ref{sect:2Dlandau}, the 2D Landau level problems are
reviewed both for the non-relativistic particles and for relativistic
particles.
The complex analyticity of the lowest Landau level wavefunctions
is presented.
In Sect. \ref{sect:3DLL}, the constructions of high-dimensional
Landau levels in 3D and 4D with explicit rotational symmetries
are reviewed.
The quaternionic analyticity of the lowest Landau level
wavefunctions, and the bulk-boundary correspondences
in terms of the Euclidean and Minkowski versions of
the Weyl equation are presented.
In Sect. \ref{sect:reduction}, we review the dimensional reductions
from the 3D and 4D Landau level problems to yield the
2D and 3D isotropic but parity-broken Landau levels.
They can be constructed by combining harmonic potentials and
linear spin-orbit couplings.
In Sect. \ref{sect:diracLL}, the high-dimensional Landau levels
of Dirac fermions are constructed, which can be viewed as
Dirac equations in the phase spaces.
They are related to gapless Dirac fermions non-minimally coupled
to background fields.
In Sect. \ref{sect:landaugauge},
high-dimensional Landau levels in the anisotropic Landau-type
gauge are reviewed.
The 4D quantum Hall responses are derived as a spatially
separated chiral anomaly.
Conclusions and outlooks are presented in Sect. \ref{sect:conclusion}.

\section{Histories of complex number and quaternion}
\label{sect:history}

\subsection{Complex number}
Complex number plays an essential role in mathematics and quantum physics.
The invention of complex number was actually related to the history of
solving the algebraic cubic equations, rather than solving
the quadratic equation of $x^2=-1$.
If one lived in the 16th century, one could simply say that
such an equation has no solution.
But cubic equations are different.
Consider a reduced cubic equation $x^3+p x+q=0$, which can be solved by using radicals.
Here is the Cardano formula,
\bea
x_1&=&c_1+c_2, ~~ x_2=c_1 e^{i\frac{2\pi}{3}}+ c_2 e^{-i\frac{2\pi}{3}},
\nn \\
x_3&=& c_1 e^{-i\frac{2\pi}{3}}+ c_2 e^{i\frac{2\pi}{3}},
\label{eq:cubic}
\eea
where
\bea
c_{1}=\sqrt[3]{-\frac{q}{2}+ \sqrt{\Delta}}, \ \ \, \ \ \,
c_{2}=\sqrt[3]{-\frac{q}{2}-\sqrt{\Delta}},
\eea
with the discriminant $\Delta=(\frac{q}{2})^2+(\frac{p}{3})^3$.
The key point of the expressions in Eq. \ref{eq:cubic} is that they involve
complex numbers.
For example, consider a cubic equation with real coefficients
and three real roots $x_{1,2,3}$.
It is purely a real problem: It starts with real coefficients and ends up
with real solutions.
Nevertheless, it can be proved that there is no way to bypass $i$.
Complex conjugate numbers appear in the intermediate steps,
and finally they cancel to yield real solutions.
For a concrete example, for the case that $p=-9$ and $q=8$,
complex numbers are unavoidable since $\sqrt{\Delta}=\sqrt{-11}$.
The readers may check how to arrive at three real roots
of  $x_{1,2,3}=1,-\frac{1}{2}\pm
\frac{\sqrt{33}}{2}$.

Once the concept of complex number was accepted, it opened up an entire
new field for both mathematics and physics.
Early developments include the geometric interpretation of complex
numbers in terms of the Gauss plane, the application of complex
numbers for two-dimensional rotations, and the Euler
formula
\bea
e^{i\theta}=\cos\theta+i\sin\theta.
\eea
The complex phase appears in the Euler formula, which is widely used
in describing  mechanical and electromagnetic waves in classic physics,
and also quantum mechanical wavefunctions.
Moreover, when a complex-valued function $f(x,y)$ satisfies
the Cauchy-Riemann condition,
\bea
\frac{\partial f}{\partial x} +i \frac{\partial f}{\partial y}=0,
\label{eq:cauchy}
\eea
it means that it only depends on $z=x+iy$ but not on $\bar z=x-iy$.
The Cauchy-Riemann condition sets up the foundation of complex analysis,
giving rise to the Cauchy integral,
\bea
\frac{1}{2\pi i} \oint \frac{1}{z-z_0} dz f(z)=f(z_0).
\eea

For physicists, one practical use of complex analysis is to
calculate loop integrals.
Certainly, its importance is well beyond this.
Complex analysis is the basic tool for many modern branches of mathematics.
For example, it gives the most elegant proof to {\it the fundamental theorem
of algebra}: An algebraic equation $f(z)=0$, i.e. $f(z)$ is
a $n$-th order polynomial, has $n$ complex roots.
The proof is essentially to count the phase winding number of $1/f(z)$
as moving around a circle of radius $R\to +\infty$, which simply
equals $n$.
On the other hand, the winding number is a topological invariant
equal to the number of poles of $1/f(z)$, or, the number of
zeros of $f(z)$.
It is also the basic tool for number theory: The Riemann hypothesis,
which aims at studying the distribution of prime numbers, is
formulated as a complex analysis problem of the distributions
of the zeros of the Riemann $\zeta(z)$-function.

Complex numbers actually are inessential in the entire
branches of classical physics.
It is well-known that the complex number description for classic waves
is only a convenience but not necessary.
The first time that complex numbers are necessary is
in quantum mechanics -- the Schr\"odinger equation,
\bea
i\hbar\partial_t  \psi=H\psi.
\eea
In contrast, classic wave equations involve $\partial_t^2$.
In fact,  Schr\"oedinger attempted to eliminate $i$ in his equation,
but did not succeed.
Hence, to a certain extent, $i$, or, the complex phase, is more
important than $\hbar$
in quantum physics.

\subsection{Quaternion and quaternoinic analyticity}
Since 2D  rotations can be elegantly described by
the multiplication of complex numbers.
It is reasonable to expect that 3D rotations could
also be described in a similar way by extending complex numbers
to include the 3rd dimension.
Simply adding another imaginary unit $j$ to construct $x+yi+zj$ does not
work, since the product of two imaginary units
$ij\neq i \neq j \neq \pm 1$.
It has to be a new imaginary unit defined as $k=ij$, and then
the quaternion is constructed as,
\bea
q=x+yi+zj+uk.
\eea
The quaternion algebra,
\bea
i^2=j^2=k^2=ijk=-1,
\label{eq:quaternion}
\eea
was invented by Hamilton in 1843 when he passed the Brougham bridge
in Dublin (See Fig. \ref{fig:quaternion}.).
He realized in a genius way that the product table of the imaginary
units cannot be commutative.
In fact, it can be derived based on Eq. \ref{eq:quaternion}
that $i$, $j$, and $k$ anti-commute with one another, {\it i.e.},
\bea
ij=-ji, \ \ \, jk=-kj, \ \ \, ki=-ik.
\eea
This is the first non-commutative division algebra discovered,
and actually it was constructed before the invention of the concept of
matrix.
In modern languages, quaternion imaginary units are isomorphic to
Pauli matrices $-i\sigma_1, -i\sigma_2, -i\sigma_3$.

\begin{figure}
\centering\psfig{file=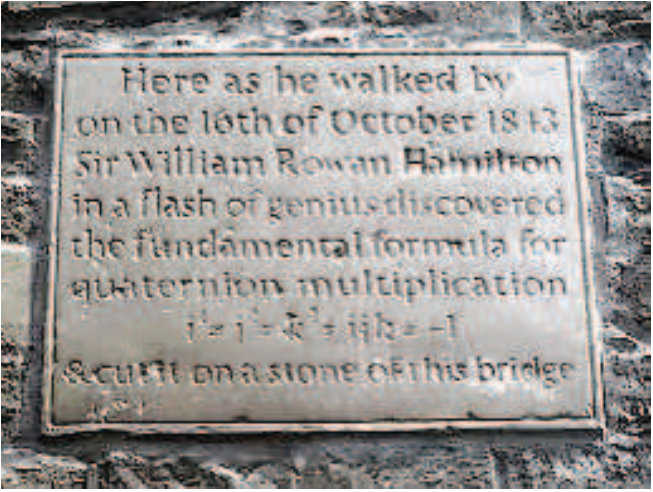,width=\linewidth, angle=0}
\caption{The quaternion plaque on Brougham Bridge, Dublin.
From wikipedia https://en.wikipedia.org/wiki/History\_of\_quaternions
}
\label{fig:quaternion}
\end{figure}

Hamilton employed quaternions to describe the 3D rotations.
Essentially he used the spin-$\frac{1}{2}$ spinor representation:
Consider a 3D rotation $R$ around the axis along the direction
of $\hat \Omega$ and the rotation angle is $\gamma$.
Define a unit imaginary quaternion,
\bea
\omega(\hat\Omega)=i\sin\theta\cos\phi+j\sin\theta\sin\phi
+k\cos\theta,
\eea
where $\theta$ and $\phi$ are the
polar and azimuthal angles of $\hat\Omega$.
Then a unit quaternion associated with such a rotation
is defined as
\bea
q=\cos\frac{\gamma}{2}+\omega(\hat\Omega)\sin\frac{\gamma}{2},
\eea
which is essentially an SU(2) matrix.
A 3D vector $\vec r$ is mapped to an imaginary quaternion
$r=xi+yj+zk$.
After the rotation, $\vec r$ is transformed to $\vec r'$,
and its quaternion form is
\bea
r'= q r q^{-1}.
\label{eq:3Drotation}
\eea
This expression defines the homomorphism from SU(2) to SO(3).
In fact, using quaternions to describe rotation is
more efficient than using the 3D orthogonal matrix,
hence, quaternions are widely used in computer graphics
and aerospace engineering even today.
If set $\vec r =\hat z$ in Eq. \ref{eq:3Drotation}, and
let $q$ run over unit quaternions, which span the $S^3$ sphere,
then a mapping from $S^3$ to $S^2$ is defined as
\bea
n=qkq^{-1},
\eea
which is the 1st Hopf map.

Hamilton spent the last 20 years of his life to promote
quaternions.
His ambition was to invent quaternion analysis which could
be as powerful as complex analysis.
Unfortunately, this was not successful because of the
non-commutative nature of quaternions.
Nevertheless, Fueter found the analogy to the Cauchy-Riemann
condition for quaternion analysis \cite{sudbery1979,frenkel2008}.
Consider a quaternionically valued function $f(x,y,z,u)$:
It is quaternionic analytic if it satisfies the following
Cauchy-Riemann-Fueter condition,
\bea
\frac{\partial f}{\partial x}
+i\frac{\partial f}{\partial y}
+j\frac{\partial f}{\partial z}
+k\frac{\partial f}{\partial u}
\label{eq:quater_ana}
=0.
\eea
Eq. \ref{eq:quater_ana} is the left-analyticity condition since
imaginary units are multiplied from the left.
A right-analyticity condition can also be similarly defined in which imaginary
units are multiplied from the right.
The left one is employed throughout this article
for consistency.
For a quaternionic analytic function, the analogy to the Cauchy integral is
\bea
\frac{1}{2\pi^2} \oiiint \frac{1}{|q-q_0|^2 (q-q_0)} Dq f(q)=f(q_0),
\label{eq:quater_cauchy}
\eea
where the integral is over a closed three-dimensional volume surrounding
$q_0$.
The measure of the volume element is,
\bea
D(q)&=&dy\wedge dz\wedge du-idx\wedge dz\wedge du \nn \\
&+& jdx\wedge dy\wedge du
-kdx\wedge dy\wedge dz,
\eea
and $K(q)$ is the four-dimensional Green's function,
\bea
K(q)=\frac{1}{q|q|^2}=\frac{x-yi-zj-uk}{(x^2+y^2+z^2+u^2)^2}.
\eea

There have also been considerable efforts in formulating quantum mechanics
and quantum field theory based on quaternions instead of complex numbers
\cite{finkelstein1962,adler1995}.
Quaternions are also used to formulate the Laughlin-like wavefunctions
in the 2D fractional quantum Hall physics \cite{balatsky1992}.

As discussed in {\it ``Selected Papers (1945-1980)
of Chen Ning Yang with Commentary"} \cite{yang2005}, C. N. Yang speculated that
quaternion quantum theory would be a major revolution to physics,
mostly based on the viewpoint of non-Abelian gauge theory.
He wrote, ``{\it... I continue to believe that the basic direction
is right.
There must be an explanation for the existence of SU(2) symmetry:
Nature, we have repeatedly learned, does not do random things at the
fundamental level.
Furthermore, the explanation is most likely in quaternion algebra:
its symmetry is exactly SU(2).
Besides, the quaternion algebra is a beautiful structure.
Yes, it is noncommutative.
But we have already learned that nature chose noncommutative algebra
as the language of quantum mechanics.
How could she resist using the only other possible nice algebra as
the language to start all the complex symmetries that she
built into the universe?"}

\section{Complex analyticity and two-dimensional Landau level}
\label{sect:2Dlandau}

In this section, I recapitulate the basic knowledge of the 2D Landau level problem,
including both the non-relativistic Schr\"odinger equation
in Sect. \ref{sect:2DSch} and the Dirac equation in Sect. \ref{sect:2DDirac}.
I explain the complex analyticity of the 2D lowest Landau level
wavefunctions.

\subsection{2D Landau level for non-relativistic electrons}
\label{sect:2DSch}

\begin{figure}
\centering\epsfig{file=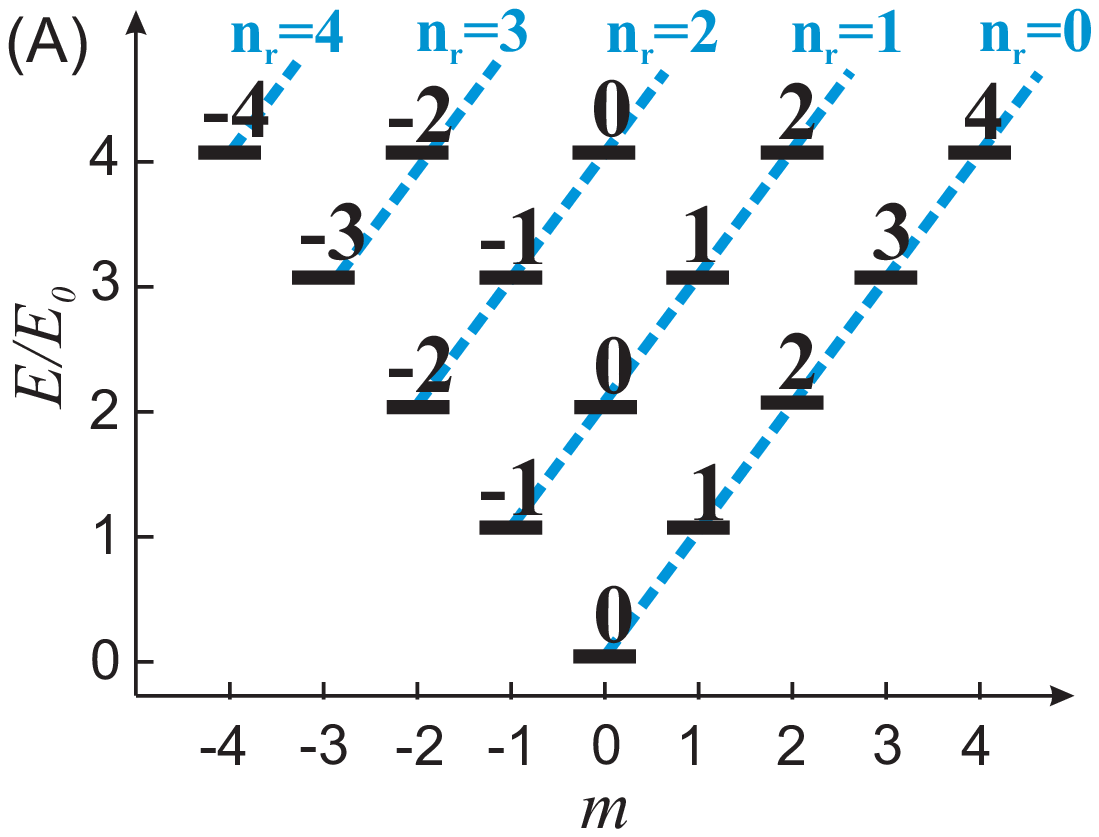,clip=1,width=0.8\linewidth,angle=0}
\hspace{10mm}
\centering\epsfig{file=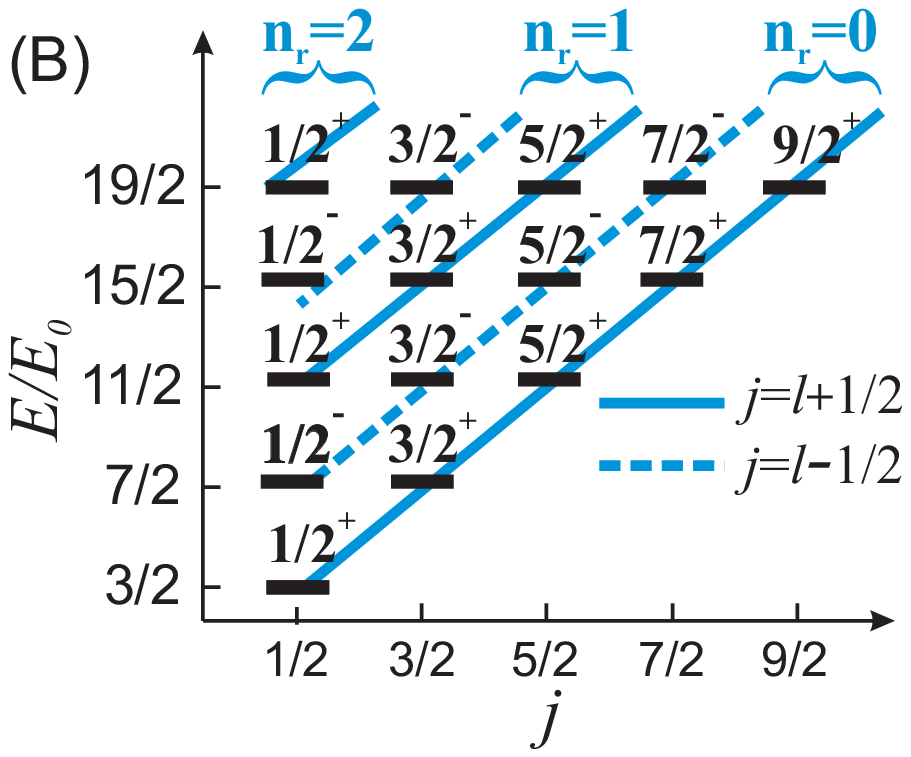,clip=1,width=0.7\linewidth,angle=0}
\caption{
A) The energy level diagram of 2D harmonic oscillators v.s. the magnetic
quantum number $m$. The states along
the tilted lines are reorganized into the 2D flat Landau levels.
B) The eigenstates of the 3D harmonic oscillator labeled by total angular
momentum $j_\pm=l\pm \frac{1}{2}$
Following the tilted solid (dashed) lines, these states are reorganized
into the 3D Landau level sates with the positive (negative) helicity
for $H^\pm_{3D,symm}$, respectively.
(From Ref.\cite{li2013}).
}
\label{fig:spectra}
\end{figure}

The reason that the 2D Landau level wavefunctions are so interesting
is that their elegancy.
The external magnetic field reorganizes harmonic oscillator wavefunctions
to yield analytic properties.
To be concrete, the Hamiltonian for a 2D electron moving in an
external magnetic field $B$ reads,
\bea
H_{2D,sym}=\frac{(\vec P -\frac{q}{c}\vec A)^2}{2M}.
\label{eq:2DLL}
\eea
In the symmetric gauge, i.e., $A_x=-\frac{1}{2}By$ and
$A_y=\frac{1}{2}Bx$, the 2D rotational symmetry is explicit.
The diamagnetic $A^2$-term gives rise to the harmonic potential, and
the cross term becomes the orbital Zeeman term.
Then Eq. \ref{eq:2DLL} can be reformulated as
\bea
H_{2D,sym}= \frac{P_x^2+P_y^2}{2M}+\frac{1}{2} M \omega_0^2 (x^2+y^2)
- \omega_0 L_z,
\label{eq:2D_symm}
\eea
where $\omega_0$ is half of the cyclotron frequency $\omega_c$;
$\omega_c=qB/(Mc)$ and $qB>0$ is assumed.
Eq. \ref{eq:2D_symm} can also be interpreted as the Hamiltonian of a
rotating 2D harmonic potential, which is how the Landau level
Hamiltonian is realized in cold atom systems.

Since these the harmonic potential and orbital Zeeman term commute
with each other, the Landau level wavefunctions are just
wavefunctions of 2D harmonic oscillators.
In Fig. \ref{fig:spectra} A), the spectra of the 2D harmonic
oscillator {\it v.s.} the magnetic quantum number $m$ are plotted,
exhibiting a linear dependence on $m$ as $E_{n_r,m}=\hbar \omega_0 (2n_r+m+1)$
where $n_r$ is the radial quantum number.
If we view this diagram horizontally, they are with finite degeneracies
exhibiting a trivial topology.
But if they are viewed along the diagonal direction, they become
Landau levels.
This reorganization is due to the orbital Zeeman term, which
also disperses linearly $E_{Z}=-m\hbar\omega_0$.
It cancels the same linear dispersion of the 2D harmonic oscillator,
such that the Landau level energies become flat.
The wavefunctions of the lowest Landau level states with $n_r=0$
are $\psi_{LLL,m}(z)=z^m e^{-|z|^2/(4l_B^2)}$ with $m\ge 0$, where
the magnetic length $l_B=\sqrt{\hbar c/(qB)}$.

Now we impose the complex analyticity, i.e., the Cauchy-Riemann condition,
to select a subset of harmonic oscillator wavefunctions.
Physically it is implemented by the magnetic field.
It just means that the cyclotron motion is chiral.
After suppressing the Gaussian factor, the lowest Landau level
wavefunction is simply,
\bea
\psi_{LLL}(z)=f(z),
\eea
which has a one-to-one correspondence to a complex analytic function.
In fact, the complex analyticity greatly facilitated the
construction of the many-body Laughlin wavefunctions \cite{laughlin1983},
\bea
\psi_L(z_1,...,z_n)=\Pi_{i<j} (z_i-z_j)^3
e^{-\sum_i\frac{|z_i|^2}{4l^2_b}},
\label{eq:laughlin}
\eea
which is actually analytic in terms of multi-complex variables.

Along the edge of a 2D Landau level system, the bulk flat states
change to 1D dispersive chiral edge modes.
They satisfy the chiral wave equation \cite{girvin1999},
\bea
\Big(\frac{1}{v_f}\frac{\partial }{\partial t} -
\frac{\partial }{\partial x}
\Big) \psi (x,t)=0,
\label{eq:chiraledge}
\eea
where $v_f$ is the Fermi velocity.

\subsection{2D Landau level for Dirac fermions}
\label{sect:2DDirac}
The is essentially a square-root problem of the Landau level
Hamiltonian of a Schr\"odinger fermion in Eq. \ref{eq:2DLL}.
The Hamiltonian reads \cite{semenoff1984},
\bea
H^{D}_{2D}=l_0\omega \Big\{(p_x - A_x) \sigma_x+(p_y -A_y) \sigma_y \Big\},
\eea
where $A_x=-\frac{1}{2}By$, $A_y=\frac{1}{2}Bx$,
$l_0=\sqrt{\frac{2\hbar c}{|qB|}}$, and
$\omega=\frac{|qB|}{2mc}$.
It can be recast in the form of
\bea
H_{2D}^{D}=\frac{ \hbar \omega}{\sqrt 2} \left[
\begin{array}{cc}
0& a_y^\dagger + ia_x^\dagger \\
a_y -i a_x & 0
\end{array}
\right],
\label{eq:2DLL_harm}
\eea
where $a_i=\frac{1}{\sqrt 2}(x_i/l_0+ip_i l_0/\hbar)$
$(i=x,y)$ are the phonon annihilation operators.

The square of Eq. \ref{eq:2DLL_harm} is reduced to the Landau level
Hamiltonian of a Schr\"odinger fermion with a supersymmetric structure
as
\bea
(H_{2D}^{D})^2/(\frac{1}{2}\hbar \omega)=\left[
\begin{array}{cc}
H_{2D,sym}-\frac{1}{2}\hbar\omega&0\\
0& H_{2D,sym}+\frac{1}{2}\hbar\omega
\end{array} \right], \nn
\\
\eea
where $H_{2D,sym}$ is given in Eq. \ref{eq:2D_symm}.
The spectra of Eq. \ref{eq:2DLL_harm} are $E_{\pm n}=\pm \sqrt n\hbar\omega$
where $n$ is the Landau level index.
The zeroth Landau level states are singled out:
Only the upper component of their wavefunctions is nonzero,
\bea
\Psi_{2D,LLL}^D(z)=\left( \begin{array}{c}
\psi_{LLL}(z)\\
0
\end{array}
\right).
\eea
$\psi_{LLL}(z)$ is the 2D lowest Landau level wavefunctions of the Schr\"odinger equation, which is complex analytic.
Other Landau levels with positive and negative energies distribute symmetrically around the zero energy.

Due to the particle-hole symmetry, each state of the zeroth Landau level
is a half-fermion Jackiw-Rebbi mode \cite{jackiw1976,heeger1988}.
When the chemical potential $\mu$ approaches $0^\pm$,
the zeroth Landau level is fully occupied, or, empty, respectively.
The corresponding electromagnetic response is,
\bea
j_\mu=\pm\frac{1}{8\pi}\frac{q^2}{\hbar}
\epsilon_{\mu\nu\lambda} F_{\nu\lambda},
\label{eq:panomaly}
\eea
known as the 2D parity anomaly
\cite{redlich1984,redlich1984a,semenoff1984,niemi1986},
where $\pm$ refer to $\mu=0^\pm$, respectively.
The two spatial components of Eq. \ref{eq:panomaly} are just the half-quantized quantum Hall conductance, and the temporal component is the half-quantized
Streda formula \cite{streda1982}.

\section{3D Landau level and quaternionic analyticity}
\label{sect:3DLL}
We have seen the close connection between complex analyticity and
2D topological states.
In this section, we discuss how to construct high-dimensional
topological states in flat spaces based on
quaternionic analyticity.

\subsection{3D Landau level Hamiltonian}
Our strategy is based on high-dimensional harmonic oscillator
wavefunctions.
Again we need to select a subset of them for non-trivial topological
properties:
{\it The selection criterion is quaternionic analyticity,
and physically it is imposed by spin-orbit coupling.}
The physical picture of the 3D Landau level wavefunctions in the symmetric-like gauge is intuitively presented in
Fig. \ref{fig:3DLL} (A).
It generalizes the fixed complex plane in the 2D Landau level
problem to a moving frame embedded in 3D.
Define a frame with the orthogonal axes $\hat e_1$, $\hat e_2$,
and $\hat e_3$, and the complex analytic wavefunctions are defined
in the $\hat e_1$-$\hat e_2$ plane with spin polarized along the
$\hat e_3$ direction.
Certainly this frame can be rotated to an arbitrary configuration.
The same strategy can be applied to any high dimensions.

\begin{figure}
\centering\epsfig{file=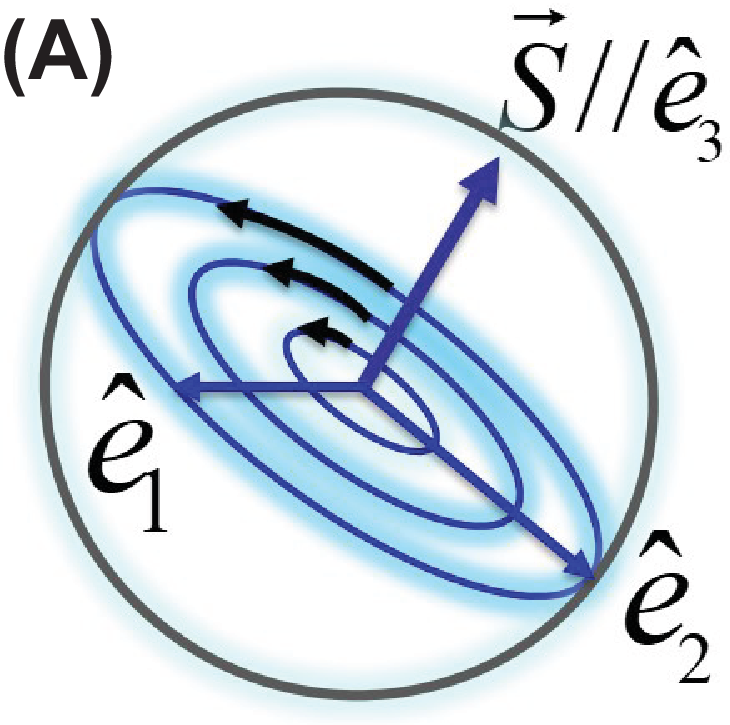,clip=1,width=0.5\linewidth,angle=0}
\hspace{20mm}
\centering\epsfig{file=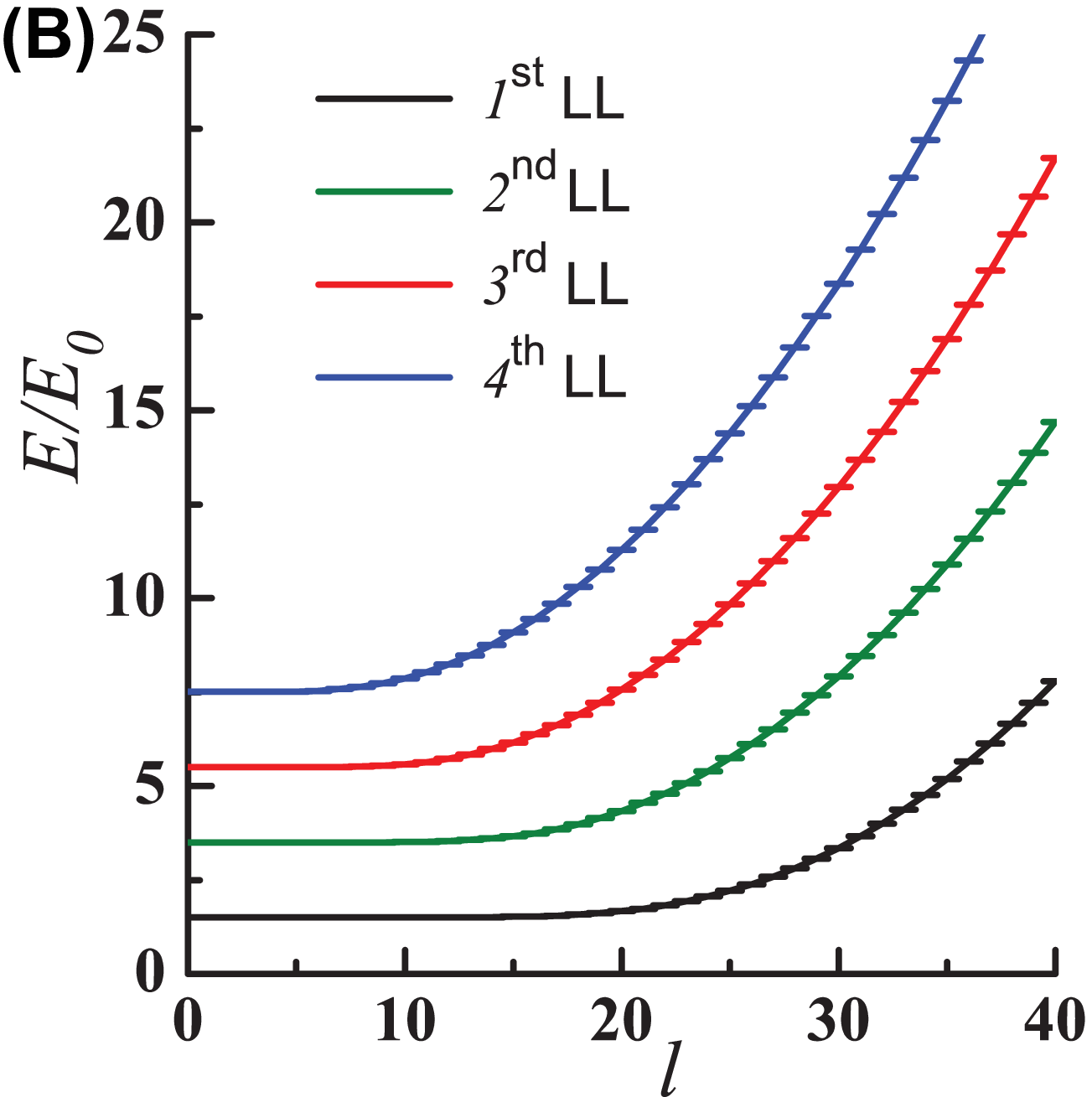,clip=1,width=0.6\linewidth,angle=0}
\caption{
A) The coherent state picture for 3D lowest Landau level
wavefunctions based on Eq. \ref{eq:3Dcoherent}.
$\hat e_1$-$\hat e_2$-$\hat e_3$ form an orthogonal triad.
The lowest Landau level wavefunction is
complex analytic in the orbital plane $\hat e_1$-$\hat e_2$
and spin is polarized along $\hat e_3$.
B) The surface spectra for the 3D Landau level Hamiltonian Eq.\ref{eq:3D_symm}.
The open boundary condition is used for a ball with the radius $R_0/l_{so}=8$.
(From Ref. \cite{li2013}.)
}
\label{fig:3DLL}
\end{figure}

Now we present the 3D Landau level Hamiltonian as constructed
in Ref. \cite{li2013}.
Consider to couple a spin-$\frac{1}{2}$ fermion to the 3D
isotropic SU(2) Aharanov-Casher potential $\vec A=\frac{G}{2} \vec \sigma
\times \vec r$ where $G$ is the coupling constant and $\vec\sigma$'s
are the Pauli matrices.
The resultant Hamiltonian is
\bea
H^{\pm}_{3D,sym}&=&\frac{1}{2M}
\big( \vec P - \frac{q}{c} \vec A (\vec{r})\big )^{2}
+V(\vec r) \nn \\
&=&\frac{P^2}{2M}+ \frac{1}{2}M\omega _{0}^{2} r^2
\mp \omega _{0}\vec{\sigma}\cdot \vec{L},
\label{eq:3D_symm}
\eea
where $\pm$  refer to $G>0 ~ (<0)$, respectively; $\omega_{0}=\frac{1}{2}\omega_{so}$ and $\omega_{so}
=|qG|/(Mc)$ is the analogy of the cyclotron frequency.
$V(r)=-\frac{1}{2}M\omega _{0}^{2}r^{2}$,
nevertheless, the $\frac{1}{2M}(\frac{q}{c})^2 A^2(r)$ term in the
kinetic energy contributes a quadratic scalar potential which
equals $2|V(r)|$, hence, Eq. \ref{eq:3D_symm} is still bound from below.
In contrast to the 2D case, $H^\pm_{3D,sym}$ preserve
time-reversal symmetry.
It can also be formulated as a 3D harmonic potential plus a
spin-orbit coupling term.
Again since these two terms commute, the 3D Landau level
wavefunctions are just the eigenstates of a 3D harmonic oscillator.

Consider the eigenstates of a 3D harmonic oscillator with an additional
spin degeneracy $\uparrow$ and $\downarrow$.
For later convenience, their eigenstates are organized into the bases of the total angular momentum $j_\pm=l\pm\frac{1}{2}$, where $\pm$ represent the positive and negative helicities, respectively.
The corresponding spectra are plotted in Fig. \ref{fig:spectra} (B),
showing a linear dispersion with respect to $l$ as
$E_{n_r,J_\pm=l\pm\frac{1}{2},J_z}=\hbar \omega_0
(2 n_r + l +\frac{3}{2})$.

Again, if we view the spectra along the diagonal direction,
the novel topology appears.
The spin-orbit coupling term $\vec \sigma \cdot \vec k$ has two branches of eigenvalues, both of which disperse linearly with $l$ as $l\hbar$ and $-(l+1)\hbar$
for the positive and negative helicity sectors, respectively.
Combining the harmonic potential and spin-orbit coupling,
we arrive at the flat Landau levels:
For $H^+_{3D}$, the positive helicity states become dispersionless
with respect to $j_+$ , a main feature of Landau levels.
Similarly, the negative helicity states become flat for $H^-_{3D}$.
States in the 3D Landau level show the same helicity.

\subsection{The SU(2) group manifold for the lowest Landau level
wavefunctions}
Having understood why the spectra are flat, now we provide
an intuitive picture for the lowest Landau level wavefunctions
with the positive helicity.
If expressed in the orthonormal basis of $(j\pm,j_z)$, they are
rather complicated,
\bea
\psi_{LLL, j_+=l+\frac{1}{2}, j_z}(r,\hat\Omega)=r^l
Y_{j_+=l+\frac{1}{2}, j_z}(\hat \Omega)
e^{-\frac{r^{2}}{4l_{so}^{2}}},
\label{eq:LL_orthonomal}
\eea
where $l_{so}=\sqrt{\hbar c/|qG|}$ is the analogy of the magnetic
length and $Y_{j_+=l+\frac{1}{2}, j_z}(\hat \Omega)$ is the
spin-orbit coupled spherical harmonic function.

Instead, they become very intuitive in the coherent state representation.
Let us start with the highest weight states with $j_+=j_z$, whose
wavefunctions are $\psi_{LLL, j_+=j_z}(r,\hat\Omega)=(x+iy)^l \exp\{-\frac{r^{2}}{4l_{so}^{2}}\} \otimes |\uparrow\rangle$.
Their spins are polarized along the $z$-direction and
orbital parts are complex analytic in the $xy$ plane.
We then perform a general SU(2) rotation such that the $xyz$-frame is rotated to the frame of $\hat e_{1}$-$\hat e_2$-$\hat e_3$.
For a coordinate vector $\vec r$, its projection in the $\hat e_1$-$\hat e_2$ plane forms a complex variable $\vec r \cdot (\hat e_1 +i \hat e_2)$
based on which we construct complex analytic functions.
Now it is clear why spin-orbit coupling is essential.
Otherwise, if the plane is flipped, then the complex variable changes to
its conjugate, and the complex analyticity is lost.
Nevertheless, since spin is polarized perpendicular to the $\hat e_1$-$\hat e_2$-plane, spin also flips during the flipping of the orbital plane, such
 that the helicity remains invariant.
In general, we can perform an arbitrary $SU(2)$ rotation on
the highest weight states and arrive at a set of coherent states
forming the over-complete bases of the lowest Landau level states as
\bea
\psi_{LLL,\hat e_{1,2,3}, j_+}(r, \hat \Omega)=[(\hat e_1 +i \hat e_2)
\cdot \vec r]^l  e^{-\frac{r^{2}}{4l_{so}^{2}}}
\otimes \ket{\alpha_{\hat e_3}},   (l\ge 0)
\nn \\
\label{eq:3Dcoherent}
\eea
where $(\hat e_3 \cdot \vec \sigma)
\ket{\alpha_{\hat e_3}}= \ket{\alpha_{\hat e_3}}$.

Now we can make a comparison among harmonic oscillator wavefunctions in different dimensions.
\begin{enumerate}
\item
In 1D, we only have the real Hermite polynomials.
\item
In 2D, a subset of harmonic wavefunctions $z^m$ (lowest Landau level) are selected exhibiting the $U(1)$ structure.
\item
In 3D, the complex plane $\hat e_1$-$\hat e_2$ associated with the
frame $\hat e_1$-$\hat e_2$-$\hat e_3$ are floating.
This is similar to the rigid-body configuration.
In other words, the configuration space of the 3D lowest Landau level
states is that of a triad, or, the $SU(2)$ group manifold.
\end{enumerate}
Since the SU(2) group manifold is isomorphic to the space of unit quaternions,
this motivates us to consider the analytic structure in terms of
quaternions, which will be presented in Sect. \ref{sect:analyticity}.

\subsection{The off-centered solutions to the lowest Landau level states}
\label{sect:off-center}

Different from the 2D Landau level Hamiltonian, which possesses the magnetic
translation symmetry, the 3D one of Eq. \ref{eq:3D_symm}
does not possess such a symmetry due to the non-Abelian nature of the
SU(2) gauge potential.
Nevertheless, based on the coherent states described by Eq. \ref{eq:3Dcoherent},
we can define magnetic translations within the
$\hat e_1$-$\hat e_2$ plane, and organize the off-centered
solutions in the lowest Landau level.

Consider all the coherent states in the $\hat e_1$-$\hat e_2$ plane
described by Eq. \ref{eq:3Dcoherent}.
We define the magnetic translation for this set of states as
\bea
T_{\hat{e}_{3}}(\vec{\delta})=\exp [-\vec \delta \cdot \vec{\nabla}+
\frac{i}{4l_{so}^{2}}
~\vec r_{12} \cdot (\hat e_3 \times \vec{\delta})],
\label{eq:tran}
\eea
where the translation vector $\vec{\delta}$ lies in the $\hat{e}_{1,2}$-plane
and $\vec{r}_{12}= \vec{r}-\hat{e}_{3}(\vec{r}\cdot \hat{e}_{3})$.
Set $\hat e_1=\hat z$, and the normal vector $\hat e_3$ lying in the
$xy$-plane with an azimuthal angle $\phi'$,
\textit{i.e.}, $\hat e_3(\phi')=\hat x\cos \phi' +\hat y\sin
\phi'$, then
$\alpha_{\hat e_3}(\phi')= \frac{1}{\sqrt 2}
(|\uparrow \rangle +e^{i\phi'} |\downarrow\rangle)$.
Consider the lowest Landau level states localized at the origin,
\bea
\psi_{l=0,\hat e_3}(r, \hat \Omega)=e^{-\frac{r^{2}}{4l_{so}^{2}}}
\otimes \ket{\alpha_{\hat e_3}},
\eea
and translate it along $\hat z$ at the distance $R$.
According to Eq. \ref{eq:tran}, we arrive at
\bea
\psi_{\phi',R} (\rho,\phi,z) =e^{i\frac{1}{2l_{so}^2} R\rho \sin (\phi-\phi')}
e^{-|\vec r- R \hat z|^2/4l_{so}^2}\otimes \alpha_{\hat e_3}(\phi'),
~
\label{eq:off-center}
\eea
where $\rho=\sqrt{%
x^2+y^2}$ and $\phi$ is the azimuthal angular of $\vec r$
in the $xy$-plane.

Now we can restore the rotational symmetry around the $\hat z$-axis by
performing the Fourier transform with respect to the angle $\phi'$, i.e.,
$\psi _{j_{z}=m+\frac{1}{2},R}(\rho,\phi,z)=\int_{0}^{2\pi }\frac{d\phi'
}{2\pi }e^{im\phi' }\psi _{\phi' ,R}$.
We arrive at the eigenstates of $j_z$ as
\bea
\psi _{j_{z}=m+\frac{1}{2},R}(\rho,\phi,z)
&=& e^{\frac{-|\vec{r}-R\hat{z}|^{2}}{4l_{so}^{2}}}
e^{im\phi }\Big\{J_{m}(x)|\uparrow \rangle \nn \\
&+&J_{m+1}(x)e^{i\phi }|\downarrow \rangle \Big\},
\label{eq:fourier}
\eea
where $x=R\rho /(2l_{so}^{2})$.
It describes a wavefunction with the shape of
an ellipsoid, whose distribution in the $xy$-plane is within the
distance of $ml_{so}^2/R$.
The narrowest states $\psi_{\pm \frac{1}{2},R}$ have an aspect
ratio scaling as $l_{so}/R$ when $R$ goes large.
On the other hand, for those states with $|m|<R/l_{so}$, they localize
within the distance of $l_{so}$ from the center located at $R\hat{z}$.
As a result, the real space local density of states of the lowest
Landau level grows linearly with $R$.

\subsection{Quaternionic analyticity of the lowest Landau
level wavefunctions}
\label{sect:analyticity}

In analogy to complex analyticity of the 2D lowest Landau level states,
we have found that the helicity structure of the 3D lowest Landau levels
leads to quaternionic analyticity.

Just like two real numbers forming a complex number, a two-component
complex spinor $\psi=(\psi_\uparrow, \psi_\downarrow)^T$ can be mapped
to a quaternion by multiplying a $j$ to the 2nd component
\bea
f=\psi_\uparrow + j \psi_\downarrow.
\label{eq:map}
\eea
Then the familiar symmetry transformations can be represented via
multiplying quaternions.
The time-reversal transformation $i\sigma_2 \psi^*$ becomes $T f=-fj$ satisfying $T^2=-1$.
The $U(1)$ phase $e^{i\theta}\to f e^{i\theta}$,
and the SU(2) rotation becomes
\bea
e^{i\frac{\phi}{2}\sigma_x} \psi \to e^{k\frac{\phi}{2}} f, \ \ \,
e^{i\frac{\phi}{2}\sigma_y} \psi \to e^{j\frac{\phi}{2}} f, \ \ \,
e^{i\frac{\phi}{2}\sigma_z} \psi \to e^{-i\frac{\phi}{2}} f.
\nn \\
\eea

To apply the Cauchy-Riemann-Fueter condition Eq. \ref{eq:quater_cauchy}
to 3D, we simply suppress the 4th coordinate,
\bea
\frac{\partial f}{\partial x}
+i\frac{\partial f}{\partial y}
+j\frac{\partial f}{\partial z}
\label{eq:quater_ana_3D}
=0.
\eea

We prove a remarkable property below
that this condition (Eq. \ref{eq:quater_ana_3D}) is rotationally invariant.
\begin{lemma}
If a quaternionic wavefunction $f(x,y,z)$ is quaternionic analytic,
i.e., it satisfies the Cauchy-Riemann-Futer condition,
then after an arbitrary rotation, the consequential wavefunction $f^\prime(x,y,z)$ remains quaternionic analytic.
\end{lemma}
\begin{proof}
Consider an arbitrary SU(2) rotation $g(\alpha,\beta,\gamma)=e^{-i\frac{\alpha}{2}\sigma_z}
e^{-i\frac{\beta}{2}\sigma_y} e^{-i\frac{\gamma}{2}\sigma_z}$,
where $\alpha, \beta,\gamma$ are Eulerian angles.
In the quaternion representation, it maps to
$g=e^{i\frac{\alpha}{2}} e^{-j\frac{\beta}{2}}
e^{i\frac{\gamma}{2}}$.
After this rotation $f(x,y,z)$ transforms to
\bea
f^{\prime}(x,y,z)
=e^{i\frac{\alpha}{2}} e^{-j\frac{\beta}{2}}
e^{i\frac{\gamma}{2}} f( x^\prime, y^\prime, z^\prime),
\label{eq:rotation}
\eea
where  $(x^\prime,y^\prime$, $z^\prime)$ are the coordinates
by applying $g^{-1}$ on $(x,y,z)$.
It can be checked that
\bea
&&\Big(\frac{\partial}{\partial x}+ i\frac{\partial}{\partial y}
+j\frac{\partial}{\partial z} \Big)
e^{i\frac{\alpha}{2}} e^{-j\frac{\beta}{2}}
e^{i\frac{\gamma}{2}}\nn \\
&=&e^{i\frac{\alpha}{2}} e^{-j\frac{\beta}{2}} e^{i\frac{\gamma}{2}}
\Big(\frac{\partial}{\partial x'} +i\frac{\partial}{\partial y'}
+j\frac{\partial}{\partial z'} \Big).
\eea
Then we have
\bea
(\frac{\partial}{\partial x}+ i\frac{\partial}{\partial y}
+j\frac{\partial}{\partial z})f'(x,y,z)=0.
\eea
Hence, the Cauchy-Riemann-Fueter condition is rotationally
invariant.
\end{proof}

Based on this lemma, we prove the quaternionic analyticity of
the 3D lowest Landau level wavefunctions.
\begin{theorem}
The 3D lowest Landau level wavefunctions of $H^+_{3D,sym}$ in Eq. \ref{eq:3D_symm}
have a one-to-one correspondence to the quaternionic analytic polynomials
in 3D.
\end{theorem}
\begin{proof}
We denote the quaternionic polynomials, which correspond to the
orthonormal bases of the lowest Landau level wavefunctions
in Eq. \ref{eq:LL_orthonomal}, as $f^{LLL}_{j_+,j_z}$ with
$j_+=l+\frac{1}{2}$, and $-j_+\le j_z \le j_+$.
The highest weight states $f^{LLL}_{j_+,j_+}=(x+iy)^l$ are
complex analytic in the $xy$-plane, hence, it is obviously
quaternionic analytic.
Since all the coherent states can be obtained from the highest weight
states via rotations, they are also quaternionic analytic.
The coherent states form a set of overcomplete basis of the lowest Landau
level wavefunctions, hence all the lowest Landau level wavefunctions
are quaternionic analytic.

Next we prove the completeness that $f^{LLL}_{j_+,j_z}$'s
form the complete basis of the quaternoinic analytic polynomials
in 3D.
By counting the degrees of freedom of the $l$-th order polynomials
of $x,y,z$, and the number of the constraints from Eq. \ref{eq:quater_cauchy},
we calculate the total number of the linearly independent
$l$-th order quaternionic analytic polynomials
as $C^2_{l+2}-C^2_{l+1}=l+1$.
On the other hand, any lowest Landau level state in the sector of $j_+=l+\frac{1}{2}$ can be represented as
\bea
f_l(x,y,z) =\sum_{m=0}^{l} f^{LLL}_{j_+=l+\frac{1}{2},j_z=m+\frac{1}{2}} q_m,
\label{eq:cmplt}
\eea
where $q_m$ is a quaternion constant coefficient.
Please note that $q_{lm}$'s are multiplied from right
due to the non-commutativity of quaternions.
In Eq. \ref{eq:cmplt}, we have taken into account the fact
$f^{LLL}_{j_+,-j_z}=-f^{LLL}_{j_+,-j_z}j$ due to the time-reversal
transformation.
Hence, the degrees of freedom in the lowest Landau level with $j_+=l+\frac{1}{2}$ is also $l+1$.
Hence, the lowest Landau level wavefunctions are complete
for quaternionic analytic polynomials.
\end{proof}

\subsection{Generalizations to 4D and above}
The above procedure can be straightforwardly generalized to four and
even higher dimensions.
To proceed, we need to employ the Clifford algebra $\Gamma$-matrices.
Their ranks in different dimensions and concrete representations
are presented in Appendix \ref{appendix:cliff}.
Then we use the $N$-D harmonic oscillator
potential combined with spin-orbit coupling as
\bea
H^{ND, LL}=\frac{p_{ND}^2}{2m}+\frac{1}{2}m \omega_0^2 r_{ND}^2
-\omega_0
\sum_{1\le i<j \le N} \Gamma_{ij} L_{ij},
\label{eq:4DQHE}
\eea
where $L_{ij}=r_i p_j - r_j p_i$.
The spectra of Eq. \ref{eq:4DQHE} were studied in the context
of the supersymmetric quantum mechanics \cite{bagchi2001}.
However, its connection with Landau levels was not noticed
there.
The spin operators in $N$-dimensions are defined as
$\frac{1}{2}\Gamma_{ij}$.

For the 4D case, the minimal representations
for the $\Gamma$-matrices are still two-dimensional.
They are defined as
\bea
\Gamma_{ij}=-\frac{i}{2}[\sigma_{i}, \sigma_{j}], \ \ \,
\Gamma_{i4}=\pm \sigma_{i},
\eea
with $1\le i<j\le3$.
The $\pm$ signs of $\Gamma^{i4}$ correspond to two complex
conjugate irreducible
fundamental spinor representations of $SO(4)$, and the $+$
sign will be taken below.
The spectra of the positive helicity states are flat as
$E_{+, n_r}=(2n_{r}+2) \hbar \omega$.
The coherent state picture for the 4D lowest Landau levels can
be similarly constructed as follows:
Again pick up two orthogonal axes $\hat e$ and $\hat f$ to form
a 2D complex plane,  and define complex analytic functions therein as,
\bea
(x_a \hat e_a + i x_a \hat f_a)^l e^{-\frac{r^2}{4l_{so}^2}}
\otimes |\alpha_{\hat e, \hat f}\rangle,
\label{eq:4Dcoherent}
\eea
where $|\alpha_{\hat e, \hat f}\rangle$ is the eigenstate
of $\Gamma^{\hat e,\hat f}=\hat e_a \hat f_b \Gamma^{ab}$ satisfying
\bea
\Gamma^{\hat e,\hat f} |\alpha_{\hat e, \hat f}\rangle
=|\alpha_{\hat e, \hat f}\rangle.
\eea
Hence, its spin is locked with its orbital angular momentum in the
$\hat e$-$\hat f$ plane.

Following similar methods in Sect. \ref{sect:analyticity}, we can prove
that the 4D lowest Landau level wavefunctions for Eq. \ref{eq:4DQHE}
satisfy the 4D Cauchy-Riemann-Futer condition Eq. \ref{eq:quater_ana},
and thus are quaternionic analytic functions.
Again it can be proved that they form the complete basis for
quaternionic left-analytic polynomials in 4D.
As for even higher dimensions, quaternions are not defined.
Nevertheless, the picture of the complex analytic function defined in
the moving frame still applies.
If we still work in the spinor representation, we can express
the lowest Landau level wavefunctions as
$\psi_{LLL}(x_i)=f_{LLL}(x_i) e^{-\frac{r^2}{2l_0^2}}$,
where each component of the spinor $f_{LLL}$ is a polynomial of
$r_i$ $(1\le i \le N)$.
To work out the analytic properties of $f_{LLL}$, we
factorize Eq. \ref{eq:4DQHE} as
\bea
H^{ND,LL}= \hbar \omega_0 \left(\Gamma^i a_i^\dagger\right)
\left(\Gamma^j a_j\right),
\eea
where $a_i$ is the phonon operator in the $i$-th dimension
defined as
$a_i=\frac{1}{\sqrt 2}\big(\frac{1}{l_0} r_i +i \frac{l_0}{\hbar }p_i\big)$,
and $l_0=\sqrt{\frac{\hbar}{m\omega_0}}$.
Then $f_{LLL}(x_i)$ satisfies the following equation,
\bea
\Gamma^j \frac{\partial}{\partial x_j} f_{LLL}(x_i)=0,
\eea
which can be viewed as the Euclidean version of the
Weyl equation.
When coming back to 3D and 4D, and following
the mapping Eq. \ref{eq:map}, we arrive at quaternionic
analyticity.

New let us construct the off-centered solutions to the lowest
Landau level states in 4D.
We use $\vec r$ to denote a point in the subspace of $x_1$-$x_2$-$x_3$,
and $\hat \Omega$ as an arbitrary unit vector in it.
Set $\hat e=\hat \Omega$ and $\hat f=\hat e_4$ (the unit vector
along the 4th axis) in Eq. \ref{eq:4Dcoherent}.
$\alpha_{\hat\Omega\hat e_4}$ satisfies
\bea
(\sigma_{i4}\Omega_i)\alpha_{\hat\Omega\hat e_4 }
=(\vec \sigma \cdot \hat \Omega) \alpha_{\hat\Omega\hat e_4}
=\alpha_{\hat\Omega\hat e_4},
\eea
hence,
\bea
\alpha_{\hat\Omega\hat e_4}=(\cos\frac{\theta}{2},\sin\frac{\theta}{2} e^{i\phi})^T,
\label{eq:alpha}
\eea
where we have used the gauge convention that the singularity
is located at the south pole.
Define the magnetic translation in the $\hat\Omega$-$\hat e_4$ plane,
\bea
T_{\hat\Omega x_4} (u_0 \hat x_4)=
\exp\Big(-u_0 \partial_{x_4} -\frac{i}{4l_{so}^2}
(\vec r \cdot \hat\Omega) u_0 \Big),
\eea
which translates along the $\hat e_4$-axis at the distance of $u_0$.
Apply this translation to the state of $e^{-r^2/4l_{so}^2}\otimes \alpha_{\hat\Omega\hat e_4}$, we arrive at the off-center solution
\bea
\psi_{\Omega,u_0}(\vec r, x_4)= e^{-\frac{r^2+x_4^2}{4l_{so}^2}}
e^{-i\frac{r u_0}{2l_{so}^2}} \otimes \alpha_{\hat\Omega\hat e_4}.
\eea
Next, we perform the Fourier transform over the direction $\hat \Omega$,
\bea
\psi_{4D;j,j_z}(\vec r, x_4)&=& \int d \Omega ~
 Y_{-\frac{1}{2},l+\frac{1}{2}, m+\frac{1}{2}} (\hat \Omega)
\psi_{\Omega,w_0}(\vec r, x_4), \ \ \,  \ \ \,
\label{eq:4D_offcenter}
\eea
where $j=l+\frac{1}{2}$ and $j_z=m+\frac{1}{2}$.
Due to the Berry phase structure $\alpha_{\hat \Omega\hat e_4}$
over $\hat \Omega$, monopole spherical harmonics,
$Y_{-\frac{1}{2},l+\frac{1}{2},m+\frac{1}{2}}(\hat\Omega)$,
are used instead of the regular spherical harmonics.
Then Eq. \ref{eq:4D_offcenter} possesses the 3D rotational symmetry
around the new center $(0,0,0, w_0)$, and Eq. \ref{eq:4D_offcenter}
possesses  with the good quantum numbers of 3D angular momentum $(j, j_z)$.
The monopole harmonic function $Y_{q;jj_z}(\hat \Omega)$ here
is defined as
\bea
Y_{q;jj_z}(\hat \Omega)=\sqrt{\frac{2j+1}{4\pi}} e^{i(j_z+q)\phi} d^l_{j_z,-q}(\theta),
\eea
where $\theta$ and $\phi$ are the polar and azimuthal angles
of $\hat \Omega$, and
$d^l_{j_z,-q}(\theta)=\langle jj_z|e^{-iJ_y\theta}|j-q\rangle$
is the standard Wigner rotation $d$-matrix.
The gauge choice is consistent with that of the
Eq. \ref{eq:alpha}.

\subsection{\small Boundary helical Dirac and Weyl modes}
\label{sect:3Dboundary}
The topological nature of the 3D Landau level states exhibits clearly
in the gapless surface spectra.
Consider a ball of the radius $R_0\gg l_{so}$ imposed by the open
boundary condition.
We have numerically solved the spectra as shown in Fig. \ref{fig:3DLL}
(B).
Inside the bulk, the Landau level spectra are flat with respect to $j_+=l+\frac{1}{2}$.
As $l$ increases to large values such that the classic orbital radiuses
approach the boundary, the Landau levels become surface states
and develop dispersive spectra.

We can also derive the effective equation for the surface mode based on Eq. \ref{eq:3D_symm}.
Since $r$ is fixed at the boundary, it becomes a rotator equation on the sphere.
By linearizing the dispersion at the chemical potential $\mu$,
and replacing the angular momentum quantum number $l$ by the operator
$\vec \sigma \cdot \vec L$, we arrive at
$H_{sf}=(v_{f}/R_{0})\vec{\sigma} \cdot \vec{L}-\mu$
with $v_{f}$ the Fermi velocity.
This is the helical Dirac equation defined on the boundary sphere.
When expanded in the local patch around the north pole
$R_{0}\hat{z}$, we arrive at
\bea
H_{sf}=\hbar v_{f}(\vec{k}\times \vec{\sigma})\cdot \hat{z}-\mu.
\eea
The gapless surface states are robust against time-reversal invariant
perturbations if odd numbers of helical Fermi surfaces exist according
to the $\mathbb{Z}_2$ criterion \cite{kane2005,kane2005a}.
Since each fully occupied Landau level contributes one helical
Dirac Fermi surface, the bulk is topologically nontrivial if
odd numbers of Landau levels are occupied.

A similar procedure can be applied to the high-dimensional case
by imposing the open boundary condition to Eq. \ref{eq:4DQHE}.
For example, around the north pole of $r_N=(0,...., R_0)$,
the linearized low energy equation for the boundary modes is
\bea
H_{bd}=\hbar v_f \sum_{i=1}^{D-1} k_i \Gamma^{iN} -\mu.
\eea
On the boundary of the 4D sphere, it becomes the 3D Weyl equation that
\bea
H_{bd}=\hbar v_f \vec k \cdot \vec \sigma -\mu.
\eea

\subsection{Bulk-boundary correspondence}

\begin{widetext}

\begin{table}[h]
\begin{center}
\begin{tabular}{|c|c|c|c|}
\hline
& Bulk (Euclidean) & Boundary (Minkowski)\\
\hline
       &                      &               \\
2D LLL &  complex analyticity &  1D chiral wave \\
              & $\partial_x f +i \partial_y f=0$
              & $\partial_t \psi+\partial_x \psi=0$     \\
       &                      &                         \\
\hline
       &                      &                         \\
3D LLL & (3D) quaternionic analyticity & 2D helical Dirac mode\\
       & $\partial_x f +i \partial_y f +j\partial_z f=0$ &
       $ \partial_t \psi +\sigma_2 \partial_x \psi -\sigma_1 \partial_y \psi=0$\\
       &                      &                         \\
\hline
       &                      &                         \\
4D LLL &  quaternionic analyticity & 3D Weyl mode \\
       & $\partial_x f +i \partial_y f +j\partial_z f +k\partial_u f=0$
       & $\partial_t \psi +\sigma_1 \partial_x \psi +\sigma_2 \partial_y \psi
       +\sigma_3 \partial_z \psi=0$     \\
       &                      &                         \\
\hline
\end{tabular}
\caption{\normalsize Bulk-boundary correspondence in the lowest 
(LLL) states in 2, 3, and 4 dimensions.}
\end{center}
\end{table}
\end{widetext}

We have already studied the bulk and boundary states of 2D, 3D and
4D lowest Landau level states.
They exhibit a series of interesting bulk-boundary correspondences
as summarized in Table I.
In the 2D case, the bulk wavefunctions in the lowest Landau level
is complex analytic satisfying the Cauchy-Riemann condition.
The 1D edge states satisfy the chiral wave equation Eq. \ref{eq:chiraledge}.
It is essentially the Weyl equation, which is actually a
single component equation in 1D.
It can be viewed as the Minkowski version of the Cauchy-Riemann
condition of Eq. \ref{eq:cauchy}.
Or, conversely, the Cauchy-Riemann condition for the bulk
wavefunctions can be viewed as the Euclidean
version of the Weyl equation.

This correspondence goes in parallel in 3D and 4D lowest Landau
level wavefunctions.
Their bulk wavefunctions satisfy the quaternionic analytic
conditions, which can be viewed as an Euclidean version of the
helical Dirac and Weyl equations, respectively.

\subsection{Many-body interacting wavefunctions}

It is natural to further investigate many-body interacting wavefunctions
in the lowest Landau levels in 3D and 4D.
As is well-known that the complex analyticity of the 2D lowest Landau level wavefunctions results in the elegant from of the 2D Laughlin wavefunction Eq. \ref{eq:laughlin}, which describes a 2D quantum liquid \cite{laughlin1983,girvin1999}.
It is natural to further expect that the quaternionic analyticity of the
3D and 4D lowest Landau levels would work as a guidance in
constructing high-dimensional SU(2) invariant quantum liquid.
Nevertheless, the major difficulty is that quaternions do not commute.
It remains challenging how to use quaternions to represent
a many-body wavefunction with the spin degree of freedom.

Nevertheless, we present below the spin polarized fractional many-body
states in 3D and 4D Landau levels.
In the 3D case, if the interaction is spin-independent, we expect
spontaneous spin polarization at very low fillings due to the
flatness of lowest Landau level states in analogy to the 2D quantum Hall ferromagnetism \cite{Lee1990,Sondhi1993,Fertig1994,Read1995,girvin1999}.
According to Eq. \ref{eq:3Dcoherent}, fermions concentrate to
the highest weight states in the orbital plane $\hat e_1$-$\hat e_2$
with spin polarized along $\hat e_3$, then it is reduced to a 2D
quantum Hall-like problem on a membrane floating in the 3D space.
Any 2D fractional quantum Hall-like state can be formed under suitable
interaction pseudopotentials \cite{haldane1983,haldane1985,Prange1990}.
For example, the $\nu=\frac{1}{3}$ Laughlin-like state on this
membrane is constructed as
\bea
&&\Psi_{\frac{1}{3}}(\vec r_1,\vec r_2, ..., \vec r_n)_
{\sigma_1\sigma_2...\sigma_n}\nn \\
&=&\prod_{i < j} [(\vec r_i-\vec r_j)\cdot (\hat e_1 +i \hat e_2)]^3
\otimes \ket{\alpha_{\hat e_3}}_{\sigma_1} \ket{\alpha_{\hat e_3}}_{\sigma_2} ...
\ket{\alpha_{\hat e_3}}_{\sigma_n}, \nn \\
\label{eq:ferro_3D}
\eea
where $\ket{\alpha_{\hat e_3}}$ represents a polarized spin eigenstate along $\hat e_3$, and the Gaussian weight is suppressed for simplicity.
Such a state breaks rotational symmetry and time-reversal symmetry spontaneously,
thus it possesses low energy spin-wave modes.
Due to the spin-orbit locked configuration in Eq. \ref{eq:3Dcoherent},
spin fluctuations couple to the vibrations of the orbital motion plane,
thus the metric of the orbital plane becomes dynamic.
This is a natural connection to the work of geometrical
description in fractional quantum Hall states \cite{haldane2011,can2014,klevtsov2017}.

Let us consider the 4D case, we assume that spin is polarized
as the eigenstate $\ket{\uparrow}$ of
$\Gamma^{12}=\Gamma^{34}=\sigma_3$.
The corresponding spin-polarized lowest Landau level wavefunctions
are expressed as
\bea
\Psi^{4D}_{LLL,m,n}=(x+iy)^m (z+iu)^n 
\otimes \ket{\uparrow},
\eea
with $m, n\ge 0$.
If all these spin polarized lowest Landau level states with
$0 \le m < N_m$ and $0 \le n < N_n $ are
filled, the many-body wavefunction is a Slater-determinant as
\bea
\Psi^{4D}(v_1, w_1; \cdots; v_N, w_N)= \det[v_i^{\alpha} w_i^{\beta}],
\label{eq:4Dwf}
\eea
where the coordinates of the $i$-th particle form two pairs
of complex numbers as $v_i=x_i+iy_i$ and $w_i=z_i+iu_i$;
$\alpha$, $\beta$ and $i$ satisfy $0 \le \alpha <N_m $,  $0\le
\beta < N_n $, and $1\le i \le N=N_m N_n$.
Such a state has a 4D uniform density as $\rho=\frac{1}{4\pi^4 l_G^2}$.
A Laughlin-like wavefunction can be written down as
$\Psi^{4D}_k =(\Psi^{4D})^k$ whose filling relative to $\rho$ should be $1/k^2$.
It would be interesting to further study its electromagnetic responses
and fractional topological excitations based on $\Psi^{4D}_k$.
Again such a state spontaneously breaks rotational symmetry, and
the coupled spin and orbital excitations would be interesting.

\section{Dimensional reductions: 2D and 3D Landau levels with broken parity}
\label{sect:reduction}

In this section, we review another class of isotropic Landau level-like
states with time-reversal symmetry but broken parity in both 2D and 3D.
The Hamiltonians are again harmonic potential plus spin-orbit coupling,
but it is the coupling between spin and linear momentum, not orbital
angular momentum \cite{wu2011_Ian,li2012a,li2016}.
They exhibit topological properties very similar to Landau levels.

An early study of these systems filled with bosons can be found in Ref. \cite{wuexciton2008}.
The spin-orbit coupled Bose-Einstein condensations (BECs) spontaneously
break time-reversal symmetry, and exhibit the skyrmion type spin textures
coexisting with half-quantum vortices, which have been reviewed
in Ref. \cite{zhouXF2013}.
Spin-orbit coupled BECs have become an active research direction
of cold-atom physics, as extensively studied in literature.
\cite{wu2011_Ian,hu2012,sinha2011,ghosh2011,wang2010,ho2011}.

\subsection{The 2D parity-broken Landau levels}
We consider the Hamiltonian of Rashba spin-orbit coupling combined
with a 2D harmonic potential as
\bea
H_{2D,hm}=-\frac{\hbar^2\nabla^2}{2M}+\frac{1}{2} M \omega^2r^2
-\lambda (-i\hbar\nabla_x \sigma_y +i\hbar \nabla_y \sigma_x),
\nn \\
\label{eq:rashba}
\eea
where $\lambda$ is the spin-orbit coupling strength with the unit of velocity.
Eq. \ref{eq:rashba} possesses the $C_{v\infty}$-symmetry and time-reversal
symmetry.

We fill the system with fermions and work on its topological properties.
There are two length scales.
The trap length scale is defined as $l_T=\sqrt{\frac{\hbar }{M\omega}}$.
If without the trap, the single particle states $\psi_\pm (\vec k)$
are eigenstates of the helicity operator $\vec \sigma \cdot (\vec k
\times \hat z)$ with eigenvalues of $\pm 1$.
Their spectra are $\epsilon_{\pm}(\vec k)=
\hbar^2(k \mp k_{0})^2/(2M)$, respectively.
The lowest energy states are $\psi_+(\vec k)$ located around a ring in
momentum space with radius $k_{0}=M\lambda/\hbar$.
This introduces a spin-orbit length scale as $l_{so}=1/k_0$.
Then the ratio between these two length scales defines a dimensionless
parameter $\alpha=l_T/l_{so}$,
which describes the spin-orbit coupling strength relative
to the harmonic potential.

In the case of strong spin-orbit coupling, {\it i.e.}, $\alpha\gg 1$,
a clear picture appears in momentum space.
The low energy states  are reorganized from the plane-wave states
$\psi_+(\vec k)$ with $k\approx k_0 $.
Since $\alpha\gg 1$, we can safely project out the high energy
negative helicity states $\psi_{-}(\vec k)$, then
the harmonic potential in the low energy sector
becomes a Laplacian in momentum space
coupled to a Berry connection $\vec A_k$ as
\bea
V=\frac{M}{2}\omega^2 r^2=\frac{M}{2} \omega^2 (i\nabla_k - A_k)^2,
\eea
which drives particle moving around the ring.
It is well-known that for the Rashba Hamiltonian, the Berry connection
$A_k$ gives rise to a $\pi$-flux at $\vec k=(0,0)$ but zero
Berry curvature at $\vec k\neq 0$  \cite{xiao2010}.
The consequence is that
the angular momentum eigenvalues become half-integers as $j_z=m+\frac{1}{2}$.
The angular dispersion of the spectra can be estimated as
$E_{agl}(j_z)=(j_z^2/2\alpha^2) \hbar\omega$, which is strongly
suppressed by spin-orbit coupling.
On the other hand, the radial energy quantization remains as
usual $E_{rad}(n_r)=(n_r+\frac{1}{2}) \hbar\omega$ up to a constant.
Thus the total energy dispersion is
\bea
E_{n_r,j_z}\approx \Big(n_r+ \frac{1}{2}
+\frac{j_z^2}{2\alpha^2} \Big )\hbar \omega.
\eea
Similar results have also been obtained in recent works of Ref.
\cite{hu2012,sinha2011,ghosh2011}.
Since $\alpha\gg 1$, the spectra are nearly flat with
respect to $j_z$, we can treat $n_r$ as a Landau level index.
The wavefunctions of Eq. \ref{eq:rashba} in the lowest Landau level
with $n_r=0$ can be expressed in the polar coordinate
as Eq. \ref{eq:rashba_wf}.

Next we define the edge modes of such systems, and
their stability problem is quite different from that of
the chiral edge modes of 2D magnetic Landau level systems.
In the regime that $\alpha\gg 1$, the spin-orbit length $l_{so}$
is much shorter than $l_T$, such that $l_T$ is viewed as the cutoff
of the sample size.
States with $|j_z|< \alpha$ are viewed as bulk states
which localize within the region of $r<l_T$.
For states with $|J_z|\sim \alpha$, their energies touch the
the bottom of the next higher Landau level, and thus they
should be considered as edge states.
Due to time-reversal symmetry, each filled Landau level of
Eq. \ref{eq:rashba} gives rise to a branch of edge modes of
Kramers' doublets $\psi_{n_r,\pm j_z}$.
In other words, these edge modes are helical rather than chiral.
Similarly to the $Z_2$ criterion in Ref. \cite{kane2005,kane2005a},
which was defined for Bloch wave states,
in our case the following mixing term,
$
H_{mx}=\psi^\dagger_{2D,n_r,j_z} \psi_{2D,n_r, -j_z}+h.c.,
$
is forbidden by time-reversal symmetry.
Consequently, the topological index for this system is $Z_2$.

\subsection {Dimensional reduction from 3D}
In fact, we construct a Hamiltonian closely related to Eq. \ref{eq:rashba}
such that its ground state is solvable exhibiting exactly flat dispersion.
It is a consequence of dimensional reduction based on the 3D Landau level Hamiltonian Eq. \ref{eq:3D_symm}.
We cut a 2D off-centered plane perpendicular to the $z$-axis
with the interception $z=z_0$.
In this off-centered plane, inversion symmetry is broken, and Eq. \ref{eq:3D_symm} is reduced to
\bea
H_{2D,re}&=& H_{2D,hm} - \omega L_z \sigma_z.
\label{eq:2D_reduce}
\eea
The first term is just Eq. \ref{eq:rashba} by identifying $\lambda=\omega z_0$
and the frequency of the 2nd term is the same as that of the harmonic trap.
If $z_0=0$, the Rashba spin-orbit coupling vanishes, and
Eq.\ref{eq:2D_reduce} becomes the 2D quantum spin-Hall Hamiltonian,
which is a double copy of Eq. \ref{eq:2D_symm}.
At $z_0\neq 0$, $\sigma_z$ is no longer conserved due to
spin-orbit coupling.

In Sect. \ref{sect:off-center}, we derived the off-centered ellipsoid
type wavefunction in Eq. \ref{eq:fourier}.
After setting $z=z_0$ in Eq. \ref{eq:fourier}, we arrive
at the following 2D wavefunction,
\bea
\psi_{2D,j_z}(r,\phi)&=&
e^{-\frac{r^2}{4l_{so}^2}}\Big\{
e^{im\phi} J_m(k_0 r) \ket{\uparrow}\nn \\
&+&e^{i(m+1)\phi} J_{m+1} (k_0 r)
\ket{\downarrow} \Big\},
\label{eq:rashba_wf}
\eea
where $J_m(k_0r)$'s are the Bessel functions.
It is straightforward to prove that the simple reduction indeed
gives rise to the solutions to the lowest Landau levels for Eq. \ref{eq:2D_reduce}, since the partial derivative along the $z$-direction
of the solution in Eq. \ref{eq:fourier} equal zero at $z=z_0$.
We also derive that the energy dispersion is
exactly flat as,
\bea
H_{2D,re} ~\psi_{2D,j_z}= \Big(1 -\frac{\alpha^2}{2} \Big)
\hbar \omega~ \psi_{2D,j_z}.
\eea

The above two Hamiltonians Eq. \ref{eq:2D_reduce} and Eq. \ref{eq:rashba} are
nearly the same except the $L_z\sigma_z$ term, whose effect
relies on the distance from the origin.
Consider the lowest Landau level solutions at $\alpha\gg 1$.
The decay length of the Gaussian factor is $l_T$.
Nevertheless, the Bessel functions peak around $k_0 r_0\approx m$, i.e., $r_0\approx \frac{m}{\alpha} l_T$.
Hence for states with $j_z<\alpha$, their wavefunctions
already decay before reach $l_T$.
Then the $L_z \sigma_z$-term compared to the Rashba one
is a small perturbation at the order of $\omega r_0/
\lambda=r_0/z_0\ll 1$.
In this regime, these two Hamiltonians are equivalent.
In contrast, in the opposite limit that $j_z\gg \alpha^2$,
the Bessel functions are cut off by the Gaussian factor,
and only their initial power-law parts participate, and
the classic orbit radiuses are just $r_0\approx \sqrt{m} l_T$,
then the physics of Eq. \ref{eq:2D_reduce} is controlled by
the $L_s\sigma_z$-term as in the quantum spin Hall systems.
For the intermediate region that $\alpha<j_z< \alpha^2$, the physics
is a crossover between the above two limits.

The many-body physics based on the above spin-orbit coupled Landau levels
in Eq. \ref{eq:rashba_wf} would be very interesting.
Fractional topological states would be expected which
are both rotationally and time-reversal invariant.
However, $s_z$ is not a good quantum number and party is also broken,
hence, these states should be very different from a double copy of
fractional Laughlin states with spin up and down particles.
The nature of topological excitations and properties of
edge modes will be deferred to a future study.

\subsection{The 3D parity-broken Landau levels}
We have also considered the problem of 3D harmonic potential
plus a Weyl-type spin-orbit coupling as
\cite{li2012a},
\bea
H_{3D,hm}=-\frac{\hbar^2 \nabla^2}{2M}+\frac{1}{2}M \omega^2r^2
-\lambda (-i\hbar\vec \nabla \cdot \vec \sigma).
\label{eq:3DSO}
\eea
The analysis can be performed in parallel to the 2D case.
In the absence of spin-orbit coupling, the low energy states of
Eq. \ref{eq:3DSO} in momentum space form a spin-orbit sphere.
The harmonic potential further quantizes the energy spectra as
\bea
E_{n_r,j,j_z}\approx \big(n_r+ \frac{1}{2}
+\frac{j(j+1)}{2\alpha^2} \big )\hbar \omega,
\eea
where $n_r$ is the Landau level index and $j$ is the total
angular momentum.
Again $j$ takes half-integer values because the Berry phase
on the low energy sphere exhibits a unit monopole structure.

Now we perform the dimensional reduction from the 4D Hamiltonian
Eq. \ref{eq:4DQHE} to 3D.
We cut a 3D off-centered hyper-plane perpendicular to the 4-th
axis with the
interception $x_4=u_0$.
Within this 3D hyper-plane of $(x_1,x_2,x_3, x_4=u_0)$,
Eq. \ref{eq:4DQHE} is reduced to
\bea
H_{3D,re}=H_{3D,hm}-\omega \vec L \cdot \vec \sigma,
\label{eq:3D_reduce}
\eea
where the first term is just Eq. \ref{eq:3DSO} with the spin-orbit
coupling strength set by $\lambda=\omega u_0$.
Again, based on the center-shifted wavefunction in the lowest Landau
level Eq. \ref{eq:4D_offcenter}, and by setting $x_4=u_0$, we
arrive at the following wavefunction
\bea
\psi_{3D, J J_z}(\vec r)&=& e^{-\frac{r^2}{4l_{so}^2}} \Big\{ j_l(k_0 r)
Y_{+,J,J_z} (\Omega_r) \nn \\
&+&i j_{l+1}(k_0 r) Y_{-,J,J_z} (\Omega_r) \Big\},
\label{eq:3D_WF}
\eea
where $k_0=u_0/l_T^2=m \lambda/\hbar$;
$j_l$ is the $l$-th order spherical Bessel function.
$Y_{\pm,j, l,j_z}$'s are the spin-orbit coupled spherical harmonics defined as
\bea
Y_{+,j,l,j_z}(\Omega)=\Big(\sqrt{\frac{l+m+1}{2l+1}} Y_{lm},
\sqrt{\frac{l-m}{2l+1}} Y_{l,m+1}\Big)^T
\nn
\eea
with the positive
eigenvalue of $l\hbar$ for $\vec \sigma \cdot \vec L$, and
\bea
Y_{-,j,l,j_z}(\Omega)=\Big(-\sqrt{\frac{l-m}{2l+1}} Y_{lm},
\sqrt{\frac{l+m+1}{2l+1}} Y_{l,m+1}\Big)^T
\nn
\eea
with the negative
eigenvalue of $-(l+1)\hbar$  for $\vec \sigma \cdot \vec L$.
It is straightforward to check that $\psi_{3D,j,j_z}(\vec r)$
in Eq. \ref{eq:3D_WF} is the ground state wavefunction
satisfying
\bea
H_{3D,re} \psi_{3D,j,j_z}(\vec r)= \Big (\frac{3}{2}
-\frac{\alpha^2}{2} \Big )\hbar \omega
\psi_{3D,j,j_z}(\vec r).
\eea

\section{High-dimensional Landau levels of Dirac fermions}
\label{sect:diracLL}

In this section, we review the progress on the study of 3D Landau levels
of relativistic Dirac fermions \cite{li2012}.
This is a square-root problem of the 3D Landau level problem of
Schr\"odinger fermions reviewed in Sect. \ref{sect:3DLL}.
This can also be viewed of Landau levels of complex quaternions.

\subsection{3D Landau levels for Dirac fermions}
In Eq. \ref{eq:2DLL_harm}, two sets of phonon creation and annihilation
operators $(a_x,a_y;a_x^\dagger,a_y^\dagger)$ are combined with the
real and imaginary units to construct Landau level Hamiltonian for
2D Dirac fermions.
Science in 3D there exist three sets of phonon creation and annihilation operators,
complex numbers are insufficient.

The new strategy is to employ Pauli matrices $\vec\sigma$
such that
\bea
H_{3D}^D&=& v \Big\{ \alpha_i  p_i
+ \gamma_i  i \hbar \frac{r_i}{l_0^2}  \Big\}
=
\frac{\hbar \omega}{\sqrt 2}
\left[
\begin{array}{cc}
0 & i\sigma_i a_i^\dagger \\
-i\sigma_i a_i & 0
\end{array}
\right ], \ \ \,
\label{eq:Dirac_3D}
\eea
where the repeated index $i$ runs over $x,y$ and $z$;
$v_F= \frac{1}{2} l_0\omega$.
The convention of $\gamma$-matrices is
\bea
\beta=\gamma_0=\tau_3\otimes I, \ \ \, \alpha_i=\tau_1\otimes\sigma_i, \ \ \,
\gamma_i=\beta\alpha_i=i\tau_2\otimes\sigma_i.
\eea
Eq. \ref{eq:Dirac_3D} contains the complex combination of
momenta and coordinates, thus it can be viewed as the generalized
Dirac equation defined in the phase space.
Apparently, Eq. \ref{eq:Dirac_3D} is rotationally invariant.
It is also time-reversal invariant with the definition
$T=\gamma_2\gamma_3K$ where $K$ is the complex conjugation,
and $T^2=-1$.
Since $\beta H_{3D}^D \beta =-H_{3D}^D$, $H_{3D}^D$ possesses the
particle-hole symmetry and its spectra are symmetric with
respect to the zero energy.

Similar to the 2D case,  $(H^{D}_{3D})^2$ has a supersymmetric structure.
The square of Eq. \ref{eq:Dirac_3D} is block-diagonal, and two blocks
are just the non-relativistic 3D Landau level Hamiltonians
in Eq. \ref{eq:3D_symm},
\bea
\frac{(H^{D}_{3D})^2}{\frac{1}{2}\hbar \omega}=
\left[ \begin{array}{cc}
H^+_{3D,sym}-\frac{3}{2}\hbar\omega& 0\\
0&H^-_{3D,sym}+\frac{3}{2}\hbar\omega
\end{array}
\right],
\label{eq:super}
\eea
where the mass $M$ in $H^\pm_{3D,sym}$ is defined through the relation $l_{0}=\sqrt{\hbar/(M\omega)}$.
Based on Eq. \ref{eq:super}, the energy eigenvalues of Eq. \ref{eq:Dirac_3D} are
$E_{\pm n_r, j, j_z}=\pm \hbar \omega \sqrt n_r$, corresponding to taking
positive and negative square roots of the non-relativistic dispersion,
respectively.
The Landau level wavefunctions of the 3D Dirac electrons are expressed
in terms of the non-relativistic ones of
Eq. \ref{eq:3D_symm} as
\bea
\Psi_{\pm n_r,j,j_z} (\vec r)=\frac{1}{\sqrt 2}
\left( \begin{array}{c}
\psi_{n_r,j_+,l, j_z} (\vec r) \\
\pm i\psi_{n_r-1,j_-, l+1, j_z }(\vec r)
\end{array}
\right).
\label{eq:LLWF}
\eea
Please note that the upper and lower two components
possess different values of orbital angular momenta.
They exhibit opposite helicities of $j_\pm$, respectively.
The zeroth Landau level ($n_r=0$) states are special: There is only one
branch, and only the first two components of the wavefunctions
are non-zero as
\bea
\Psi_{n_r=0, j, j_z} (\vec r)=\left[\begin{array}{c}
\Psi_{LLL,j_+,j_z}(\vec r)\\
0
\end{array}
\right],
\eea
where $\Psi_{LLL,j_+,j_z}$'s are the lowest Landau level solutions
to the non-relativistic Hamiltonian Eq. \ref{eq:LL_orthonomal}.

Again the nontrivial topology of the 3D Dirac Landau problem manifests
in the gapless surface modes.
Consider a spherical boundary with a large radius $R$.
The Hamiltonian takes the form of Eq. \ref{eq:Dirac_3D} inside the sphere,
and has the usual massive Dirac Hamiltonian
$H_{D}=\alpha_i P_i + \beta\Delta$ outside.
We also take  the limit of $|\Delta|\rightarrow \infty$.
Loosely speaking, this is a square-root version of the open boundary
problem of the 3D non-relativistic case
in Sect. \ref{sect:3Dboundary}.
Since square-roots can be taken as positive and negative,
each branch of the surface modes in the non-relativistic Schr\"odinger case
corresponds to a pair of relativistic surface branches.
These two branches disperse upward and downward as increasing the
angular momentum $j$, respectively.
However, the zeroth Landau level branch is singled out.
We can only take either the positive, or, negative
square root, for its surface excitations.
Hence, the surface spectra connected to the bulk zeroth Landau level
disperse upward or downward depending on the sign of the vacuum mass.

\subsection{Non-minimal Pauli coupling and anomaly}
Due to the particle-hole symmetry of Eq. \ref{eq:Dirac_3D}, the 3D
zeroth Landau level states are half-fermion modes in the same
way as those in the 2D Dirac case.
Moreover, in the 3D case, the degeneracy is over the
3D angular momentum numbers $(j_+, j_z)$, thus the
degeneracy is much higher than that of 2D.
According to whether the chemical potential $\mu$ approaches
$0^+$ or $0^-$, each state in the zeroth lowest Landau level contributes a
 positive, or, negative half fermion number, respectively.
The Lagrangian of the 3D massless Dirac Landau level problem is,
\bea
L= \bar \psi \Big\{ \gamma_0 i\hbar \partial_t
- i v \gamma_i  \hbar \partial_i\Big\} \psi
- v \hbar  \bar\psi  i \gamma_0 \gamma_i \psi
F^{0i}(r),
\label{eq:lang}
\eea
where $F^{0i}= x_i/l_0^2$.
In all the dimensions higher than 2, $i\gamma_0\gamma_i$'s are
a different set from $\gamma_i$'s, thus Eq. \ref{eq:lang} is an example of
non-minimal coupling of the Pauli type.
More precisely, it is a coupling between the electric field and the electric
dipole moment.
In the 2D case, the Lagrangian has the same form as Eq. \ref{eq:lang},
however, since $\gamma_{0,1,2}$ are just the usual Pauli matrices,
it is reduced to the minimal coupling to the $U(1)$ gauge field.

Eq. \ref{eq:lang} is a problem of massless Dirac fermions coupled to
a background field via non-minimal Pauli coupling at 3D and above.
Fermion density is pumped by the background field from vacuum.
This is similar to parity anomaly, and indeed it is
reduced to parity anomaly in 2D.
However, the standard parity anomaly only exists in even spatial
dimensions \cite{redlich1984,redlich1984a,semenoff1984,niemi1986}.
By contrast, the Landau level problems of massless Dirac fermions
can be constructed in any high spatial dimensions.
Obviously, they are not chiral anomalies defined in odd spatial
dimensions, either.
It would be interesting to further study the nature of such
kind of ``anomaly''.

In fact, Eq. \ref{eq:Dirac_3D} is just one possible representation
for Landau levels of 3D massless Dirac fermions.
A general 3D Dirac Landau level Hamiltonian with a mass term
can be defined as
\bea
H_{3D}^D (\hat e_1, \hat e_2,\hat e_3)&=&
v \Big [(\vec \tau \cdot \hat e_1) \otimes\sigma_i  P_i
+ \hbar/l_0^2 (\vec \tau \cdot \hat e_2 )\otimes \sigma_i r_i
\Big] \nn \\
&+& mv^2 (\vec \tau \cdot \hat e_3)\otimes I,
\label{eq:3D_Dirac_general}
\eea
where $\tau_{1,2,3}$ are Pauli matrices acting in the particle-hole
channel, and $\hat e_{1,2,3}$ form an orthogonal triad in the 3D space.
Eq. \ref{eq:Dirac_3D} corresponds to the case of $\hat e_1=\hat x$
and $\hat e_2=\hat y$, and $m=0$.
The parameter space of $H_{3D}^D (\hat e_1, \hat e_2, \hat e_3)$
is the triad configuration space of $SO(3)$.

Consider that the configuration of the triad $\hat e_{1,2,3}$
is spatially dependent.
The first term in Eq. \ref{eq:3D_Dirac_general} should be symmetrized
as $\frac{1}{2}\vec \tau \cdot
\big[(\hat e_1(r) P_i + P_i \hat e_1(r) \big] \otimes \sigma_i$.
The spatial distribution of the triad of $\hat e_{1,2,3}(\vec r)$
can be in a topologically nontrivial configuration.
If the triad is only allowed to rotate around a fixed axis, its
configuration space is $U(1)$ which can form a vortex line
type defect.
There should be a Callan-Harvey type effect of the fermion
zero modes confined around the vortex line \cite{Callan1985}.
In general, we can also have a 3D skyrmion type defect of the triad
configuration.
These novel defect problems and the associated zero energy fermionic
excitations will be deferred for later studies.

\subsection{Landau levels for Dirac fermions in four dimensions
and above}
The Landau level Hamiltonian for Dirac fermions can be generalized
to arbitrary $N$-dimensions ($N$-D) by replacing the Pauli matrices in Eq. \ref{eq:Dirac_3D} with the Clifford algebra $\Gamma$-matrices in $N$-D
as presented in Appendix \ref{appendix:cliff}.

In odd dimensions $D=2k+1$, we use the $k$-th rank
$\Gamma$-matrices to construct the $D=2k+1$ dimensional
Dirac Landau level Hamiltonian,
\bea
H^D_{2k+1}=\frac{\hbar\omega_0}{2}\left(
\begin{array}{cc}
0& i\Gamma_i^{(k)} a^\dagger_i \\
-i\Gamma_i^{(k)} a_i & 0
\end{array}
\right),
\label{eq:highDirac}
\eea
where $\Gamma_i^{(k)}$ is $2^k\times 2^k$ dimensional matrix,
and $1\le i\le 2k+1$.
Again, $(H_{2k+1}^{D})^2$ are reduced to a supersymmetric version
of the $2k+1$-dimensional Landau level Hamiltonian for Sch\"odinger
fermions in Eq. \ref{eq:4DQHE}.
All other properties are parallel to the 3D case explained before.

For even dimensions $D=2k$, we still take Eq. \ref{eq:highDirac}
by suppressing the $2k+1$-th dimension.
Nevertheless, such a construction is reducible.
In the representation presented in
Appendix \ref{appendix:cliff},
Eq \ref{eq:highDirac} after eliminating the $\Gamma^{(k)}_{2k+1}$ term
can be factorized into a pair of Hamiltonians
\bea
H^{\pm,D}_{2k}=\frac{\hbar\omega_0}{2}\left(
\begin{array}{cc}
0& \pm a^\dagger_{2k}+ i \sum_{i=1}^k \Gamma_i^{(k-1)} a^\dagger_i \\
\pm a_k -i \sum_{i=1}^k \Gamma_i^{(k-1)} a_i & 0
\end{array}
\right),\nn \\
\eea
where $\pm$ correspond to the pair of fundamental and anti-fundamental
spinor representations in even dimensions.

For example, for the 4D system, we have
\bea
H_{4D}^{\pm, D} =\frac{\hbar \omega}{\sqrt 2}\left[ \begin{array}{cc}
0& \pm a^\dagger_4 + i \sigma_i a_i^\dagger\\
\pm a_4 - i \sigma_i  a_i &0
\end{array}
\right].
\label{eq:reduce}
\eea
Since three quaternionic imaginary units $i,j$, and $k$ can be
mapped to Pauli matrices $i\sigma_1, i\sigma_2$, and $i\sigma_3$,
respectively,
and the annihilation and creation operators are essentially complex.
$\pm a_4 - i \sigma_i  a_i$ can be viewed as complex quaternions.
Hence, Eq. \ref{eq:reduce} is a complex quaternionic
generalization of the 2D Dirac Landau level Hamiltonian Eq. \ref{eq:2DLL_harm}.

\section{High-dimensional Landau levels in the Landau-like gauge}
\label{sect:landaugauge}

We have discussed the construction of Landau levels in high dimensions
for both Schr\"odinger and Dirac fermions in the symmetric-like
gauge.
In those problems, the rotational symmetry is explicitly maintained.
Below we review the construction of Landau levels in the Landau-like
gauge by reorganizing plane-waves to exhibit non-trivial
topological properties \cite{li2013a}.
It still preserves the flat spectra but not the rotational
symmetry.

\subsection{Spatially separated 1D chiral modes -- 2D Landau level}
We recapitulate the Landau level in the Landau gauge.
By setting $A_x=By$ and $A_y=0$ in the Hamiltonian Eq. \ref{eq:2DLL},
we arrive at
\bea
H_{2D,L}&=& \frac{P_y^2}{2M} +\frac{(P_x-\frac{e}{c} A_x)^2}{2M}\nn \\
&=&\frac{P_y^2}{2M} +\frac{1}{2}M \omega^2 (y-l_B^2 P_x)^2,
\label{eq:2D_LL_Landau}
\eea
with $l_B=\sqrt{\frac{\hbar}{M\omega}}$.
The Landau level wavefunctions are a product of a plane wave along
the $x$-direction and a 1D harmonic oscillator wavefunction in the
$y$-direction,
\bea
\psi_n(x,y)=e^{ik_x} \phi_n(y-y_0(k)),
\eea
where $\phi_n$ is the $n$th harmonic oscillator eigenstate
with the characteristic length $l_B$, and its
equilibrium position is determined by the momentum $k_x$,
$y_0(k_x)=l_B^2 k_x$.

Hence, the Landau level states with positive and negative values of $k_x$
are shifted oppositely along the $y$-direction, and become spatially
separated.
If imposing the open boundary condition along the $y$-axis, chiral
edge modes appear.
The 2D quantum Hall effect is just the spatially separated 1D chiral
anomaly in which the chiral current becomes the transverse charge current.
After the projection to the lowest Landau level, we identify
$y=l_B^2 k_x$, hence, the two spatial coordinates
$x$ and $y$ become non-commutative as  \cite{lee2004}
\bea
[x,y]_{LLL}=il_B^2.
\eea
In other words, the $xy$-plane is equivalent to the 2D phase space of a 1D system
$(x;k_x)$ after the lowest Landau level projection.

\subsection{Spatially separated 2D helical modes - 3D Landau level}
The above picture can be generalized to the 3D Landau level states:
We keep the plane-wave modes with the good momentum numbers $(k_x,k_y)$
and shift them along the $z$-axis.
Spin-orbit coupling is introduced to generate the helical structure
to these plane-waves, and the shifting direction
is determined by the sign of helicity.
To be concrete, the 3D Landau level Hamiltonian in the Landau-like
gauge is constructed as follows \cite{li2013a},
\bea
H^{\pm}_{3D,L} &=&
\frac{\vec P^2}{2M}+\frac{1}{2}M \omega_{so}^2 z^2
\mp \omega_{so} z (P_x \sigma_y-P_y \sigma_x) \nn \\
&=& \frac{P_z^2}{2M } + \frac{1}{2} M \omega_{so}^2
[z \mp \frac{1}{\hbar}l_{so}^2 (P_x \sigma_y - P_y\sigma_x) ]^2, \nn\\
\label{eq:3D_Landau}
\eea
where $l_{so}=\sqrt{\hbar/(M\omega_{so})}$.

The key of Eq. \ref{eq:3D_Landau} is the $z$-dependent Rashba spin-orbit
coupling, such that it can be decomposed into a set of 1D harmonic
oscillators along the $z$-axis coupled to 2D helical plane-waves.
Define the helicity operator $\hat \Sigma_{2d} (\hat k_{2d} )=
\hat k_x \sigma_y -\hat k_y \sigma_x$ where $\hat k$ is the
unit vector along the direction of $\vec k$.
$\chi_\Sigma(\hat k_{2d})$ is the eigenstate of $\hat \Sigma$
and $\Sigma=\pm 1$ is the eigenvalue.
Then the 3D Landau level wavefunctions are expressed as
\bea
\Psi_{n, \vec k_{2d},\Sigma}(\vec{r}) = e^{i \vec k_{2d} \cdot \vec r_{2d}}
\phi_n [z - z_0( k_{2d},\Sigma)] \otimes \chi_{\Sigma}(\hat k_{2d}),
\label{eq:3DLL_WF}
\eea
where $\vec k_{2d}=(k_x,k_y)$, $\vec r_{2d}=(x,y)$,
and $k_{2d}=(k_x^2+k_y^2)^{\frac{1}{2}}$.
The energy spectra of Eq. \ref{eq:3DLL_WF} is flat as $E_n=(n+\frac{1}{2}) \hbar
\omega_{so}$.
The center of the oscillator wavefunction in Eq. \ref{eq:3DLL_WF}
is shifted to $z_0=l_{so}^2 k_{2d} \Sigma$.

The 3D Landau level wavefunctions of Eq. \ref{eq:3DLL_WF} are spatially
separated 2D helical plane-waves along the $z$-axis.
As shown in Fig. \ref{fig:3Ddemon} (A), for states
with opposite helicity eigenvalues, their central positions
are shifted in opposite directions.
If open boundaries are imposed perpendicular to the $z$-axis, each Landau
level contributes a branch of gapless helical Dirac modes.
For the system described by $H^+_{3D,L}$, the surface Hamiltonian is
\bea
H_{bd}= \pm v_f (\vec p \times \vec \sigma) \cdot \hat z-\mu,
\eea
where $\pm$ apply to upper and lower boundaries, respectively.

Unlike the 2D case in which the symmetric and Landau gauges are
equivalent,
the Hamiltonian of the symmetric-like gauge Eq. \ref{eq:3D_symm} and
that of the Landau-like gauge Eq. \ref{eq:3D_Landau} are {\it not}
gauge equivalent.
The Landau-like gauge explicitly breaks the 3D rotational symmetry
while the symmetric-like gauge preserves it.
Physical quantities calculated based on Eq. \ref{eq:3D_Landau}, such as
density of states, are not 3D rotationally symmetric as those
from Eq. \ref{eq:3D_symm}.
Nevertheless, these two Hamiltonians belong to the same topological class.

\begin{figure}[htbp]
\centering\epsfig{file=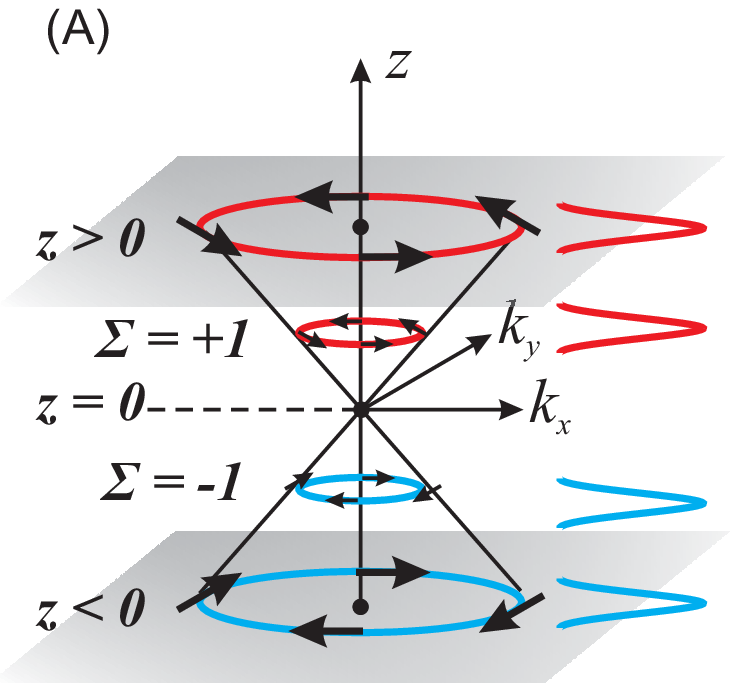,clip=1, width=0.3\textwidth,
}
\hspace{2mm}
\centering\epsfig{file=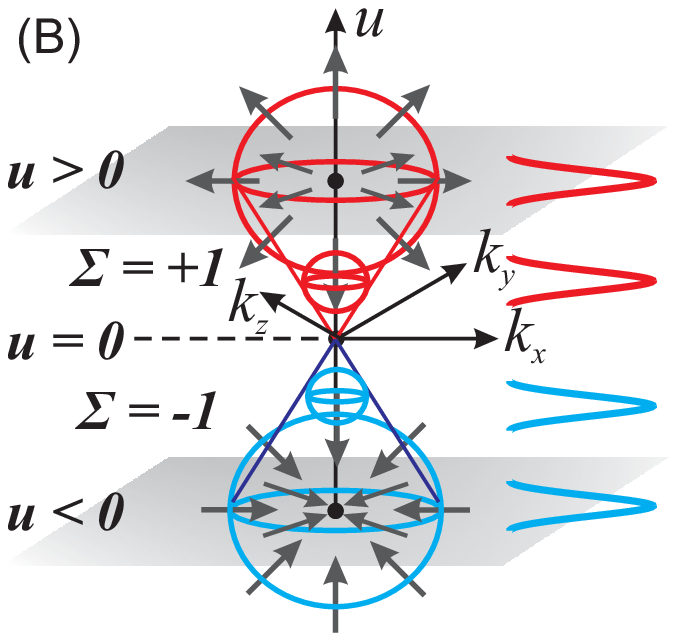,clip=2, width=0.3\textwidth,
}
\hspace{2mm}
\centering\epsfig{file=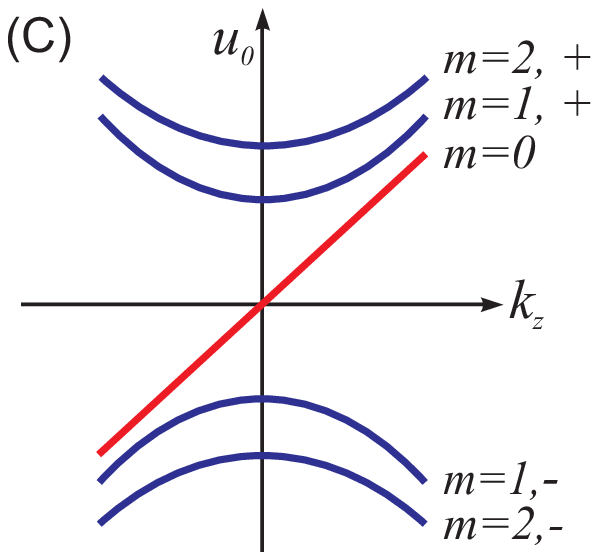,clip=1,width=0.3\textwidth, angle=0}
\caption{\small
(A) 3D Landau level wavefunctions as spatially separated 2D helical Dirac modes
localized along the $z$-axis.
(B) 4D Landau level wavefunctions as spatially separated 3D Weyl modes
localized along the $u$-axis.
Note that 2D plane-wave modes with opposite helicities
and the 3D ones with opposite chiralities are located at
opposite sides of $z=0$ and $u=0$ planes, respectively.
C) The central positions $u_{0}(m,k_z,\nu)$ of the 4d Landau levels
in the presence of the magnetic field $\vec B= B\hat z$.
The branch of $m=0$ runs across the entire $u$-axis, which gives
rise to quantized charge transport along $u$-axis in the presence
of $\vec E\parallel \vec B$ as indicated in Eq. \ref{eq:4DQHE}.
From Ref. \cite{li2013a}.
}
\label{fig:3Ddemon}
\end{figure}


\subsection{Spatially separated 3D Weyl modes --4D Landau level}
\label{sect:4D}
Again we can easily generalize the above procedure to any dimensions.
For example, in four dimensions, we need to use the 3d helicity
operator $\hat \Sigma_{3d} =\hat P_{3d} \cdot \vec \sigma$, whose
eigenstates are denoted as $\chi_{\Sigma}$ with the eigenvalues
$\Sigma=\pm 1$.
Then the 4D Landau level Hamiltonian is defined as \cite{li2013a}
\bea
H^{4d,\mp}_{LL} &=& \frac{P_u^2+\vec P_{3d}^2}{2M}+\frac{1}{2}M \omega^2 u^2
\mp \omega u \vec{P}_{3d} \cdot \vec{\sigma} \nn \\
&=& \frac{P_u^2}{2M } + \frac{1}{2} M \omega_{so}^2
(u \mp \frac{1}{\hbar}l_{so}^2 \vec P_{3d} \cdot \vec \sigma )^2, \ \ \,
\label{eq:4D_LL}
\eea
where $u$ and $P_u$ are the coordinate and momentum
in the 4th dimension, respectively,
and $\vec P_{3d}$ is defined in the $xyz$-space.
Inside each Landau level, the spectra are flat with respect to
$\vec k_{3d}$ and $\Sigma$.
Similarly to the 3D case, the 4D LL spectra and wavefunctions are solved
by reducing Eq. \ref{eq:4D_LL} into a set of 1D harmonic
oscillators along the $u$-axis as
\bea
\Psi_{n, \vec{k}_{3d},\Sigma}(\vec{r},u) = e^{i \vec{k}_{3d} \cdot \vec{r}}
\phi_n [u - u_0(k_{3d},\Sigma)] \otimes \chi_{\Sigma}(\vec k_{3d}).
\eea
The central positions $u_0(k_{3d},\Sigma)=\Sigma l_{so}^2 k_{3d}$.
This realizes the spatial separation of the 3D Weyl fermion modes
with the opposite chiralities as shown in Fig. \ref{fig:3Ddemon} (B).
With an open boundary imposed along the $u$-direction,
the 3D chiral Weyl fermion modes appear on the boundary
\bea
H_{bd}=\pm v_f (\vec k_{3D} \cdot \vec \sigma)-\mu.
\eea

\subsection{\normalsize Phase space picture of high-dimensional
Landau levels}
\label{sect:phase}
For the 2D case described by Eq. \ref{eq:2D_LL_Landau},
the $xy$-plane is equivalent to the 2D phase space of a 1D system
$(x;k_x)$ after the lowest Landau level projection.
The discrete step of $k_x$ is $\Delta k_x=2\pi/L_x$, and
the momentum cutoff of the bulk state is determined by $L_y$
as $k_{bk}=L_y/(2l_B^2)$.
Since $|k_x|<k_{bk}$, the number of states $N_{2D,LL}$ scales with
$L_xL_y$ as the usual 2D systems,
but the crucial difference is that enlarging $L_y$ does not change
$\Delta k_x$ but instead increases $k_{bk}$.

Similarly, the 3D Landau level states (Eq. \ref{eq:3D_Landau}) can be viewed as
states in the 4D phase space ($xy;k_xk_y$).
The $z$-axis plays the double role of $k_x$ and $k_y$.
After the lowest Landau level projection,
$z$ is equivalent to $z=l_{so}^2(p_x \sigma_y - p_y\sigma_x)/\hbar$,
and thus
\bea
&&[x, z]_{LLL}=i l_{so}^2 \sigma_y, \ \ \,
[y, z]_{LLL}=-i l_{so}^2 \sigma_x,
\nn \\
&&[x,y]_{LLL}=0.
\eea
The momentum cutoff of the bulk state is determined as
$(k_x^2+k^2_y)^{\frac{1}{2}}<k_{bk}=\hbar L_z/
(2 l_{so}^2)$, thus the total number of states $N$ scales as $L_x L_y L_z^2$.
As a result, the 3D local density of states linearly diverges as
$\rho_{3D}(z)\propto |z|/l_{so}^4$ as $|z|\rightarrow \infty$.
Similar divergence also occurs in the symmetric-like gauge
as $\rho_{3D}(r)\propto r/l_{so}^4$.
Now this seeming pathological result can be understood as the
consequence of squeezing states of 4D phase space $(xy;k_xk_y)$ into the 3D
real space $(xyz)$.
In other words, the correct thermodynamic limit should be taken according
to the volume of 4D phase space.
This reasoning is easily extended to the 4D LL systems
(Eq.\ref{eq:4D_LL}), which can be understood as a 6D phase
space of $(xyz;k_xk_yk_z)$.

\subsection{\normalsize Charge pumping and the 4D quantum Hall effects}
The above 4D Landau level states presented in Sect. \ref{sect:4D} exhibit
non-linear electromagnetic response \cite{zhang2001,qi2008,Werner2012,
Frohlich2000} as the 4D quantum Hall effect.
We apply the electromagnetic fields as
\bea
\vec{E}=E \hat{z}, \ \ \, \vec{B}=B \hat{z},
\eea
to the 4D Landau level Hamiltonian Eq. \ref{eq:4D_LL}
by minimally coupling fermions to the $U(1)$ vector potential,
\bea
A_{em,x}=0, \ \ \, A_{em,y} = B x, \ \ \, A_{em,z} = -cEt.
\eea
The $\vec B$-field further quantizes the chiral plane-wave modes
inside the $n$-th 4D spin-orbit Landau level states into a series of
2D magnetic Landau level states in the
$xy$-plane as labeled by the magnetic Landau level index $m$.
For the case of $m=0$, the eigen-wavefunctions are spin polarized as
\bea
\Psi_{n,m=0}(k_y,k_z)&=&e^{ik_y y+ik_z z} \phi_{n}(u-u_0(k_z,m=0))
\nn \\
&\times&\varphi_{m=0}(x-x_0(k_y))
\otimes \ket{\uparrow},
\label{eq:em4dwf}
\eea
where $\phi_n$ is the $n$-th order harmonic oscillator wavefunction with
the spin-orbit length scale $l_{so}$, and $\varphi_0$ is the
zeroth order harmonic oscillator wavefunction
with the magnetic length scale $l_B$.
The central positions of the $u$-directional and $x$-directional
oscillators are
\bea
x_0(k_y)=l_B^2 k_y, \ \ \, u_0(k_z,m=0)=l_{so}^2 k_z,
\eea
respectively.
The key point is that $u_0(k_z,m=0)$ runs across the entire $u$-axis.
In contrast, wavefunctions $\Psi_{n,m}$ with $m\ge 1$
also exhibit harmonic oscillator
wavefunctions along the $u$-axis.
However, their central positions at $m\ge 1$ are,
\bea
u_0(k_z)=\pm l_{so}^2 \sqrt{k_z^2+\frac{2m}{l_B^2}},
\eea
which only lie in half of the $u$-axis as
shown in Fig. \ref{fig:3Ddemon} (C).

Since $k_z$ increases with time in the presence of $E_z$,
$u_0(m,k_z(t))$ moves along the $u$-axis.
Only the $m=0$ branch of the magnetic Landau level states contribute
to the charge pumping since their centers go across the entire
$u$-axis, which results in an electric current along the $u$-direction.
Since $k_z(t)=k_z(0)-\frac{eE}{\hbar}t$,
during the time interval $\Delta t$, the number of electrons
passing the cross-section at a fixed $u$ is
\bea
\Delta N=\frac{L_xL_y}{2\pi l_B^2}
\frac{e E_z \Delta t }{2\pi \hbar/L_z}
=\frac{e^2}{4\pi^2\hbar^2c} \vec E \cdot \vec B V\Delta t,
\eea
where $V$ is the 3D cross-volume.
Then the current density is calculated as
\bea
j_u= n_{occ} \frac{e\Delta N}{V\Delta t}= n_{occ} \alpha \frac{e}{4 \pi^2 \hbar} \vec{E} \cdot \vec{B},
\label{eq:4Dpump}
\eea
where $\alpha$ is the fine-structure constant,
and $n_{occ}$ is the occupation number of the 4D spin-orbit Landau levels.

Eq. \ref{eq:4Dpump} is in agreement with
results from the effective field theory \cite{qi2008} as the 4D generalization
of the quantum Hall effect.
If we impose the open boundary condition perpendicular to the $u$-direction,
the above charge pump process corresponds to the chiral anomalies of Weyl
fermions with opposite chiralities on two opposite 3D boundaries, respectively.
Since they are spatially separated, the chiral current corresponds to the
electric current along the $u$-direction.


\section{Conclusions and outlooks}
\label{sect:conclusion}

I have reviewed a general framework to construct Landau
levels in high dimensions based on harmonic oscillator wavefunctions.
By imposing spin-orbit coupling, their spectra are reorganized to exhibit
flat dispersions.
In particular, the lowest Landau level wavefunctions in 3D and 4D in the quaternion representation satisfy the
Cauchy-Riemann-Fueter condition, which is the generalization
of complex analyticity to high dimensions.
The boundary excitations are the 2D helical Dirac surface modes,
or, the 3D chiral Weyl modes.
There is a beautiful bulk-boundary correspondence that the
Cauchy-Riemann-Fueter condition and the helical Dirac (chiral Weyl)
equation are the Euclidean and Minkowski representations of the
same analyticity condition, respectively.
By dimensional reductions, we constructed a class of Landau levels
in 2D and 3D which are time-reversal invariant but parity breaking.
The Landau level problem for Dirac fermions is a square-root problem of
the non-relativistic one, corresponding to complex quaternions.
The zeroth Landau level states are a flat band of half-fermion
Jackiw-Rebbi zero modes.
It is at the interface between condensed matter and high energy
physics, related to a new type of anomaly.
Unlike parity anomaly and chiral anomaly studied in field theory
in which Dirac fermions are coupled to gauge fields through
the minimal coupling, here Dirac fermions are coupled to background
fields in a non-minimal way.

I speculate that high-dimensional Landau levels could provide a platform
for exploring interacting topological states in high dimensions - due
to the band flatness, and also the quaternionic analyticity of lowest
Landau level wavefunctions.
It would stimulate the developments of various theoretical and
numerical methods.
This would be an important direction in both condensed matter physics
and mathematical physics for studying high
dimensional topological states for both
non-relativistic and relativistic fermions.
This research also provides interesting applications of quaternion
analysis in theoretical physics.


\section{Acknowledgments}
I thank Yi Li for collaborations on this set of works on high-dimensional topological states and for bringing in interesting concepts including the quaternionic analyticity.
I also thank J. E. Hirsch for stimulating discussions, and S. C. Zhang,
T. L. Ho, E. H. Fradkin, S. Das Sarma, F. D. M. Haldane, and C. N. Yang for their warm encouragements and appreciations.

\appendix
\section{Brief review on Clifford algebra}
\label{appendix:cliff}
In this part, we review how to construct anti-commutative $\Gamma$-matrices.
The familiar group is just the $2\times 2$ Pauli matrices, i.e., rank-1.
The rank-$k$ $\Gamma$-matrices can be defined recursively based on
the rank-$(k-1)$ ones.
At each level, there are $2k+1$ anti-commutative matrices, and their
dimensions are $2^k\times 2^k$.
In this article, we use the following representation,
\bea
\Gamma^{(k)}_i&=&\left[
\begin{array}{cc}
0& \Gamma_a^{(k-1)}\\
\Gamma_a^{(k-1)}& 0
\end{array}
\right], \ \ \,
\Gamma^{(k)}_{2k}=\left[
\begin{array}{cc}
0& -i I \\
iI & 0
\end{array}
\right], \nn \\
\Gamma^{(k)}_{2k+1}&=&
\left[
\begin{array}{cc}
I& 0 \\
0& -I
\end{array}
\right],
\label{eq:clifford}
\eea
where $i=1,..., 2k-1$.

In $D=2k+1$-dimensional space, the $SO(2k+1)$ fundamental spinor is $2^k$-dimensional.
The generators are constructed $S_{ij}=\frac{1}{2}\Gamma_{ij}^{(k)}$
where
\bea
\Gamma^{(k)}_{ij}=-\frac{i}{2}[\Gamma^{(k)}_i, \Gamma_j^{(k)}].
\eea

In the $D=2k$-dimensional space, there are two irreducible fundamental
spinor representations for the $SO(2k)$ group, both of which
are with $2^{k-1}$-dimensional.
Their generators are denoted as $S_{ij}$ and $S_{ij}^\prime$, respectively,
which can be constructed based on both rank-$(k-1)$
$\Gamma_i^{(k-1)}$ and $\Gamma_{ij}^{(k-1)}$-matrices.
For the first $2k-1$ dimensions, the generators share the same form
as that of the $SO(2k-1)$ group,
\bea
S_{ij}=S^\prime_{ij}=\frac{1}{2}\Gamma_{ij}^{(k-1)}, \ \ \, (1\le i < j \le 2k-1).
\label{eq:so2k_1}
\eea
Other generators $S_{i,2k}$ and $S_{i,2k}^\prime$
differ by a sign -- they are represented by the $\Gamma^{(k-1)}_i$
matrices,
\bea
S_{i,2k}=S^\prime_{i,2k}=
\pm \frac{1}{2}\Gamma_{i}^{(k-1)}, \ \ \, (1\le i \le 2k-1).
\label{eq:so2k_2}
\eea

\bibliography{TI,topo,wu_pub}


\end{document}